\documentclass[english]{article}
\usepackage[T1]{fontenc}
\usepackage[latin9]{inputenc}
\usepackage{geometry}
\usepackage{array}
\usepackage{float}
\usepackage{amsmath}
\usepackage{amssymb}
\usepackage{graphicx}

\makeatletter

\providecommand{\tabularnewline}{\\}
\floatstyle{ruled}
\newfloat{algorithm}{tbp}{loa}
\providecommand{\algorithmname}{Algorithm}
\floatname{algorithm}{\protect\algorithmname}

\newenvironment{lyxlist}[1]
{\begin{list}{}
{\settowidth{\labelwidth}{#1}
 \setlength{\leftmargin}{\labelwidth}
 \addtolength{\leftmargin}{\labelsep}
 }}
{\end{list}}

\@ifundefined{date}{}{\date{}}
\usepackage{hyperref}
\usepackage{algorithm}
\usepackage{algorithmic}
\usepackage{amsmath}
\usepackage{mathtools}
\usepackage{amsthm}
\usepackage{breqn}
\usepackage{setspace}

\newtheorem{theorem}{Theorem}[section]
\numberwithin{theorem}{section}

\newtheorem{definition}[theorem]{Definition}

\newtheorem{lemma}[theorem]{Lemma}

\newtheorem{corollary}[theorem]{Corollary}

\makeatother

\usepackage{babel}
\begin{document}
\pagenumbering{gobble}
\title{Distributed Symmetry-Breaking with Improved Vertex-Averaged Complexity
\footnote{This research has been supported by ISF grant 724/15, and by the Open University of Israel research fund.}
} 
\author{ 
Leonid Barenboim\\ 
Open University of Israel\\ 
Department of Computer Science\\ 
Ra'anana\\ 
Israel\\ 
leonidb@openu.ac.il 
\and 
Yaniv Tzur\\ 
Open University of Israel\\ 
Department of Computer Science\\ 
Ra'anana\\ 
Israel\\ 
yaniv.tzur@outlook.com 
} 
\maketitle
\begin{onehalfspacing}
\begin{abstract}
\phantomsection
\addcontentsline{toc}{section}{Abstract}We study the distributed message-passing model in which a communication
network is represented by a graph $G=(V,E)$. Usually, the measure
of complexity that is considered in this model is the worst-case complexity,
which is the largest number of rounds performed by a vertex $v\in V$.
While often this is a reasonable measure, in some occasions it does
not express sufficiently well the actual performance of the algorithm.
For example, an execution in which one processor performs $r$ rounds,
and all the rest perform significantly less rounds than $r$, has
the same running time as an execution in which all processors perform
the same number of rounds $r$. On the other hand, the latter execution
is less efficient in several respects, such as energy efficiency,
task execution efficiency, local-neighborhood efficiency and simulation
efficiency. Consequently, a more appropriate measure is required in
these cases. Recently, the vertex-averaged complexity was proposed
by \cite{Feuilloley2017}. In this measure, the running time is the
worst-case sum of rounds of communication performed by all of the
graph's vertices, averaged over the number of vertices. Feuilloley
\cite{Feuilloley2017} showed that leader-election admits an algorithm
with a vertex-averaged complexity significantly better than its worst-case
complexity. On the other hand, for $O(1)$-coloring of rings, the
worst-case and vertex-averaged complexities are the same. This complexity
is $O\left(\log^{*}n\right)$ \cite{Feuilloley2017}. It remained
open whether the vertex-averaged complexity of symmetry-breaking in
general graphs can be better than the worst-case complexity. We answer
this question in the affirmative, by showing a number of results with
improved vertex-averaged complexity for several different variants
of the vertex-coloring problem, for the problem of finding a maximal
independent set, for the problem of edge-coloring, and for the problem
of finding a maximal matching.
\end{abstract}
\pagebreak{}

\tableofcontents{}

\listoftables

\listoffigures

\listof{algorithm}{List of Algorithms}

\pagebreak{}

\pagenumbering{arabic}

\section{Introduction}

\subsection{Computational Model\label{subsec:computationalModel}}

We operate in the static, synchronous message passing model of distributed
computation. In this model, an $n$ -vertex graph $G=(V,E)$ is given,
where the vertex set \emph{$V$} represents the set of processors
and the edge set \emph{E} represents the set of communication lines
between processors in the underlying network. All processors operate
in parallel, where they can pass messages of unbounded size to their
neighbors in constant time, over the communication lines.

Also, we assume, in the static model, that no addition or deletion
of vertices or edges is performed. Every algorithm in this model operates
in synchronous rounds, where all vertices start simultaneously in
round 0, and each vertex $v\in V$ starts the $i+1$-th round only
after all vertices have finished the $i$-th round.

The most important problems studied in this model include the problems
of vertex coloring, edge coloring, finding a maximal independent set
(also known as MIS) and finding a maximal matching (also known as
MM). (See Section \ref{sec:Preliminaries} for the definitions of
these problems.) Many studies have been conducted since the mid 80's
to try and find an efficient distributed solution to these problems.
Notable early studies include \cite{Awerbuch1989,Cole1986,Goldberg1988,Israel1986,Israeli1986,Linial1987,Luby1993,Luby1986}.

\subsection{Motivation \label{subsec:Motivation}}

According to the worst-case complexity measure, the running time of
an algorithm is the number of rounds of the processor that is the
last one to terminate. Even if $n-1$ vertices terminate after just
a single round, and the remaining $n$-th vertex performs $n$ rounds,
the running time is $n$. According to this measure, exactly the same
running time is achieved in the scenario where each of the $n$ processors
perform $n$ rounds. The former scenario, however, is significantly
better in several respects. First, the overall energy consumption
may be up to $n$ times better. This is because the most significant
energy waste occurs during processor activity and communication. On
the other hand, once a processor terminates, it does not communicate
any more, and does not perform local computations. Consequently, in
a network of processors that are fed by a common energy source, the
former scenario may be $n$ times more energy efficient, although
it has the same worst-case running time as the latter one.

Another advantage of the former scenario is in improving the running
time itself, for a majority of the network processors. One example
of improving the running time of most of the network's vertices, is
a task that consists of a pair of subtasks $\mathcal{A},\mathcal{B}$
that are executed one after another. It would be better to execute
the second task in each processor once it terminates, rather than
waiting for all processors to complete the first task. This may result
in asynchronous start of the second task, which requires more sophisticated
algorithms, but significantly improves the running times of the majority
of processors. We note that if the average running time per vertex
$\bar{T}_{\mathcal{A}}$ of the first task is asymptotically smaller
than its worst-case running time $T_{\mathcal{A}}$, that is, $\bar{T}_{\mathcal{A}}=o\left(T_{\mathcal{A}}\right)$,
then the suggested execution is indeed advantageous. The reason is,
that the previous statement entails that most vertices run the first
task for $o\left(T_{A}\right)$ rounds.

Yet another example of the benefit of the former scenario is in simulations
of large-scale networks. These are commonly used for Big Data Analysis,
as well as for efficiently executing algorithms on big graphs \cite{Barenboim2017}.
In such algorithms, a distributed execution of a large-scale network
is simulated by a smaller number of processors, or just by a single
processor in some cases. Consequently, a single processor has to simulate
all rounds of a large number of vertices. In this case, minimizing
the \emph{sum of the number of rounds} of all vertices is much more
important than minimizing the maximum number of rounds of a vertex.
By minimizing this sum, the overall number of rounds a processor has
to simulate may be significantly smaller. Consequently, a complexity
measure that takes into account the sum of rounds is of great interest.
Specifically, the \emph{vertex-averaged complexity }is the sum of
rounds of all $n$ processors divided by $n$.

\section{Related Work\label{subsec:relatedResults}}

Research on the efficient distributed solution of the above-mentioned
central graph-theoretic problems conducted since the 1980's and until
a few years ago has focused almost exclusively on the analysis of
the time complexity of the developed algorithms in the \emph{worst-case
scenario}. On the other hand, during the past few years, studies presenting
new distributed algorithms for the solution of the above-mentioned
problems with no improvement in the worst-case scenario, but a significant
improvement in the average-case scenario have been published (\cite{Parter2016,Feuilloley2017}).
The analysis of the time complexity of such algorithms in the average-case
scenario is based on one of several different models.

In \cite{Feuilloley2017} (and also, in a later published extended
version \cite{Feuilloley2017a}), the model is static, that is, the
vertices and edges of the input graph do not change over time. Also,
the running time of a vertex for a certain algorithm $\mathcal{A}$
is defined in one of two ways, which the author shows to be equivalent.
In the first definition, a vertex chooses an output (its part of the
solution for the problem we are trying to solve) after some number
of rounds, based on communication with its neighbors. After selecting
the output, the vertex can continue to transmit messages and perform
computations, but cannot change its output. In the second definition,
each vertex starts with no knowledge about other vertices, and in
the $i$-th round knows the inputs and identities of its $i$-neighborhood.
At some round, the vertex chooses an output and terminates.

The average running time of an algorithm is computed by summing up
the running time of all vertices in the input graph and dividing the
sum by the number of vertices in the graph. This measure of average
running time is called by the author the \emph{complexity of an ordinary
node}, or \emph{node-averaged complexity. }More formally, let there
be a graph $G$, an algorithm $\mathcal{A}$ and let $\mathcal{ID}$
denote the set of legal ID assignments. The above-mentioned measure
of average running time is computed as follows:
\[
\bar{T}(G)=\max_{I\in\mathcal{ID}}\frac{1}{n}\sum_{v\in V(G)}r_{G,I,\mathcal{A}}(v)
\]
where the following notation has been used:
\begin{lyxlist}{00.00.0000}
\item [{$\bar{T}(G)$}] The average running time of the inspected distributed
algorithm for a graph whose number of vertices is denoted by $n$.
\item [{$V(G)$}] The set of vertices of the graph $G$.
\item [{$r_{G,I,\mathcal{A}}(v)$}] The number of rounds until a vertex
$v$ terminates in the execution of algorithm $\mathcal{A}$ on $G$
with ID assignment $I$.
\end{lyxlist}
In \cite{Parter2016}, the underlying network is assumed to be dynamic.
That is, we begin with an initial graph and an initial solution to
some problem. Over a time period that consists of discrete units,
vertices and edges may be added or deleted from the initial graph.
This may require correcting the solution of the problem we wish to
solve. In \cite{Parter2016}, the average time complexity of algorithms
running in this model is analyzed using amortized time analysis.

Some important results were obtained in \cite{Parter2016} in the
dynamic model as well. However, in this paper, we will focus on the
static model which is generally incomparable with the dynamic model.
Therefore, we will not elaborate on the results of \cite{Parter2016}
any further.

In \cite{Feuilloley2017}, the author mainly studies the average time
complexity of algorithms on cycles and other specific sparse graphs.
For the problem of leader election on cycles, the author showed the
following positive result. There is an exponential gap between the
vertex-averaged complexity, which is $O(\log n)$, and the worst-case
time complexity, which is $O(n)$. However, for other problems, such
as 3-coloring a cycle, the author shows that the vertex-averaged complexity
cannot be improved. The author also generalized this lower bound to
a class of sparse graphs he called \emph{q-sparse graphs.} 

This paper employs the first definition of \cite{Feuilloley2017}
with a slight difference. Once a vertex has finished executing an
algorithm and has decided upon a final output, it sends the final
output once to all its neighbors and terminates completely. Afterwards,
the vertex performs no further local computation or communication
in subsequent rounds. We name the measure of time complexity obtained
by taking the mentioned first definition of \cite{Feuilloley2017},
with the mentioned difference, the \emph{vertex-averaged complexity}
of an algorithm.

\section{Goal}

Several different results were obtained by the author of \cite{Feuilloley2017}.
One such result was negative, showing that the vertex-averaged complexity
and worst-case complexity of $3$-coloring rings are asymptotically
the same. Another result, on the other hand, was positive, showing
that the vertex-averaged complexity of leader election is $O(\log n)$
rounds, while in the worst-case, solving the same problem requires
$\Omega\left(n\right)$ rounds. However, the following question remained
open for general graphs, and for sparse graphs not contained in the
class of $q$-sparse graphs studied in \cite{Feuilloley2017}. The
question is, for the above-mentioned symmetry-breaking problems (vertex-coloring,
maximal independent set, edge-coloring, maximal matching), whether
an improvement can be achieved in terms of vertex-averaged complexity.
This is the subject of this paper.

\section{Overview of Our Results and Comparison with Previous Work \label{sec:comparisonWithPreviousWorkSection}}

Our first result is an $O(a)$-forests-decomposition with $O(1)$
vertex-averaged complexity. The\emph{ arboricity} $a$ is the minimum
number of forests that the edge-set of a graph can be partitioned
into. For many important graph families, such as planar graphs, graphs
that exclude any fixed minor, and graphs with constant genus, the
arboricity is bounded by a constant. For graphs with constant arboricity,
we devise deterministic algorithms for $O(a\cdot\log^{*}n)$-coloring,
MIS, maximal matching and $\left(2\Delta-1\right)$-edge coloring
with a vertex-averaged complexity of $O\left(\log^{*}n\right)$. It
was either shown or implied in \cite{Barenboim2008}, that any algorithm
for $O(a)$-forests-decomposition or $O(a\cdot\log^{*}n)$-vertex-coloring,
for a constant arboricity $a$, requires $\Omega(\log n)$ rounds,
or $\Omega\left(\frac{\log n}{\log\left(\log^{*}n\right)}\right)$,
respectively, in the worst-case. Therefore, our results demonstrate
that various problems, including coloring, have vertex-averaged complexity
that is significantly lower than the best possible worst-case complexity.
This is interesting, in view of the result of \cite{Feuilloley2017}
that coloring rings (that have arboricity 2) with a constant number
of colors has the same vertex-averaged and worst-case complexities.
Moreover, prior to our work, the best known vertex-averaged complexities
of the discussed problems were the same as the worst-case ones. Thus,
our results significantly improve the performance of the algorithms.

Our results also apply to non constant arboricity. In this case, for
deterministic algorithms, for the above-mentioned problems, we have
an additive term between $\tilde{O}(\sqrt{a})$ and $O(a)$ in the
vertex-averaged running time. (See tables \ref{tab:vertexColComparisonTable},
\ref{tab:otherProblemsComparisonTable} for a detailed description
of the running times of the various algorithms of ours.) We note,
however, that our algorithms do not incur an additional factor of
$O(\log n)$ in the running time. On the other hand, this factor is
unavoidable in the worst-case analysis in the problems of $O(a)$-forests-decomposition
and $O(a\cdot\log^{*}n)$-coloring. Computing an $O(a)$-forests-decomposition
requires $\Omega\left(\frac{\log n}{\log a}\right)$ rounds in the
worst case. Computing an $O\left(a\cdot\log^{*}n\right)$-vertex-coloring
requires $\Omega\left(\frac{\log n}{\log\left(\log^{*}n\right)+\log a}\right)$
rounds in the worst-case. We also note, that for $a=O(\log^{*}n)$,
$O\left(a\cdot\log^{*}n\right)$ coloring requires $\Omega\left(\frac{\log n}{\log\left(\log^{*}n\right)}\right)$
rounds in the worst-case \cite{Barenboim2008}. (We elaborate further
on the reason the last lower bound holds for $a=O(\log^{*}n)$ in
Section \ref{sec:discussionAndConclusionsSection}).

In addition, we devise deterministic vertex coloring algorithms that
are a function of $a,n$ and possibly also of a parameter $k$. One
such algorithm computes an $O(a^{2}\log n)$-coloring in $O(1)$ vertex-averaged
complexity. Another algorithm computes an $O(ka^{2})$-coloring in
$O\left(\log^{\left(k\right)}n\right)$ vertex-averaged complexity,
where $\log^{(k)}n$ is the iterated logarithmic function for $k$
iterations. For an appropriate value of $k$ (described in Sections
\ref{subsec:O(k * alpha)VertexColoringInImprovedVertexAveragedComplexity},
\ref{subsec:O(ka^2)Coloring}), we obtain $O\left(a^{2}\cdot\log^{*}n\right)$-vertex-coloring
with vertex-averaged complexity $O\left(\log^{*}n\right)$. This last
result is interesting due to the fact that the worst-case time complexity
of the best currently-known algorithm for $O\left(a^{2}\right)$-vertex-coloring
is $O\left(\log n\right)$. Also, the best-currently known algorithm
for $O\left(a^{2}\right)$-vertex-coloring, for constant arboricity,
is the best possible one, in the worst case \cite{Barenboim2008}.
The newly devised algorithm colors a graph with nearly the same number
of colors, with a vertex-averaged complexity that is substantially
smaller.

Yet another algorithm devised in the paper, deterministically colors
an input graph using $O\left(ka\right)$ colors, with vertex-averaged
complexity $O\left(a\log^{\left(k\right)}n\right)$.

Even more generally, for graphs with maximum degree $\Delta$, we
devise a scheme that reduces a worst-case solution with $f\left(\Delta,n\right)$
time to a solution with a vertex-averaged complexity of $O\left(f\left(a,n\right)\right)$.
This works for any problem that satisfies that a partial solution
to the problem can be extended to a complete solution. (See Section
\ref{sec:Solving-Problems-of-Extension-From-Any-Partial-Solution}
for more details.) We note that the aforementioned graph families
with constant arboricity have unbounded maximum degree. Consequently,
the resulting running time that depends on $a$, rather than on $\Delta$,
is significantly better. Moreover, we do not incur additional logarithmic-in-$n$
factors if they are not present in $f\left(\Delta,n\right)$. This
is in contrast to worst-case analysis which necessarily incurs $O(\log n)$
factors in some problems, as explained above. 

While all above-mentioned results are deterministic, with a guaranteed
vertex-averaged complexity bound, we also devise some randomized algorithms
in which the bounds hold with high probability. (This probability
is $1-\frac{1}{n^{c}}$, for an arbitrary large constant $c$.) In
particular, our analysis of $(\Delta+1)$-coloring of general graphs
results in a vertex-averaged complexity of $O(1)$, with high probability.
In addition, we compute $O\left(a\log\log n\right)$-coloring with
$O(1)$ vertex-averaged complexity, which is of interest when the
arboricity is significantly smaller than the maximum degree.

Our results are obtained using structures that eliminate in each round
or phase a constant fraction of vertices from further participation
in the algorithm. These vertices finalize their results, and terminate.
Consequently, the number of active vertices in each round or phase
decays exponentially as the algorithm proceeds. Consequently, the
vertex-averaged number of rounds or phases is $O(1)$. (This is because
the sum of rounds or phases performed by all vertices during the algorithm
is $O(n)$.) In particular, the number of vertices that are active
for a long period of time is very small, roughly $O\left(\frac{n}{2^{i}}\right)$,
where $i$ denotes the number of rounds or phases.

Tables \ref{tab:vertexColComparisonTable}-\ref{tab:otherProblemsComparisonTable}
compare previous work with results obtained in this paper. We note
the following:
\begin{itemize}
\item The time bound of each randomized algorithm holds with high probability.
\item The time complexity of all results of previous work is a worst-case
time complexity measure. However, no algorithms with better vertex-averaged
complexities have been previously devised, to the best of our knowledge. 
\item The time complexity of results of this paper presented in the tables
is a vertex-averaged complexity measure. 
\item The parameters $\delta$,$\epsilon$ and $\eta$ denote an arbitrarily
small positive constant. 
\item ``(Det.)'' denotes ``Deterministic'' and ``(Rand.)'' denotes
``Randomized''.
\end{itemize}
\begin{table}[H]
\begin{centering}
\begin{tabular}{|>{\raggedright}p{3.5cm}|>{\raggedright}p{7cm}|>{\raggedright}p{4.6cm}|}
\hline 
\textbf{Number of colors } & \textbf{Our vertex-averaged time} & \textbf{Previous running time}

\textbf{(worst case)}\tabularnewline
\hline 
$O(ka)$ & $O\left(a\log^{\left(k\right)}n\right)$ (Det.)\\
($2\le k\le\rho\left(n\right)$, $\rho\left(n\right)=O\left(\log^{*}n\right)$
is defined in Section \ref{subsec:O(k * alpha)VertexColoringInImprovedVertexAveragedComplexity})  & $O\left(a\log n\right)$ (Det.) \cite{Barenboim2008}\tabularnewline
\hline 
$O(a\log^{*}n)$ & $O\left(a\log^{*}n\right)$ (Det.) & $O\left(a\log n\right)$ (Det.) \cite{Barenboim2008}\tabularnewline
\hline 
$O(a^{1+\eta})$ & $O(\log a\log\log n)$ (Det.) & $O\left(\log a\log n\right)$ (Det.) \cite{Barenboim2011}\tabularnewline
\hline 
$O(a^{2}\log n)$ & $O(1)$ (Det.) & $O\left(\frac{\log n}{\log a+\log\log n}\right)$ (Det.) \cite{Barenboim2008}\tabularnewline
\hline 
$O(ka^{2})$ & $O\left(\log^{\left(k\right)}n\right)$ (Det.)\\
($2\le k\le\rho\left(n\right)$, $\rho\left(n\right)=O\left(\log^{*}n\right)$
is defined in Section \ref{subsec:O(k * alpha)VertexColoringInImprovedVertexAveragedComplexity})  & $O(\log n)$ (Det.) \cite{Barenboim2008}\tabularnewline
\hline 
$O\left(a^{2}\cdot\log^{*}n\right)$ & $O\left(\log^{*}n\right)$ (Det.) & $O(\log n)$ (Det.) \cite{Barenboim2008}\tabularnewline
\hline 
$\Delta+1$ & $O\left(\sqrt{a}\log^{2.5}a+\log^{*}n\right)$ (Det.) & $O\left(\sqrt{\Delta}\log^{2.5}\Delta+\log^{*}n\right)$\\
(Det.) \cite{Fraigniaud2016}\tabularnewline
\hline 
$\Delta+1$ & $O(1)$ (Rand.) & $2^{O\left(\sqrt{\log\log n}\right)}$ (Rand.) \cite{Chang2018}\tabularnewline
\hline 
$O\left(a\log\log n\right)$ & $O(1)$ (Rand.) & $\Omega\left(\frac{\log n}{\log a+\log\log\log n}\right)$ (Det.)
\cite{Barenboim2008}\\
$O\left(a\log n\right)$ (Det.) \cite{Barenboim2008}\\
$O\left(\log^{2}n\right)$ (Rand.) \cite{Barenboim2008}\cite{Barenboim2013}\tabularnewline
\hline 
\end{tabular}
\par\end{centering}
\caption{Comparison of vertex-coloring algorithms.\label{tab:vertexColComparisonTable}}
\end{table}
\begin{table}[H]
\begin{centering}
\begin{tabular}{|>{\raggedright}p{4.2cm}|>{\raggedright}p{3.7cm}|>{\raggedright}p{6.3cm}|}
\hline 
\textbf{Problem} & \textbf{Our vertex-averaged time} & \textbf{Previous running time}

\textbf{(worst case)}\tabularnewline
\hline 
Maximal Independent Set,\\
$(2\Delta-1)$-edge-coloring,\\
Maximal Matching & $O\left(a+\log^{*}n\right)$ (Det.) & $2^{O(\sqrt{\log n})}$ (Det.) \cite{Panconesi1996}\\
$O\left(\frac{\log n}{\log\log n}\right)$ (for $a\le\log^{\frac{1}{2}-\delta}n$)
(Det.) \cite{Barenboim2010}\\
$O\left(a+\log n\right)$ \\
(for MM and $(2\Delta-1)$-edge-coloring) \\
(Det.) \cite{Barenboim2010,Barenboim2009}\\
$O(a+a^{\epsilon}\log n)$ (for MIS) (Det.) \cite{Barenboim2011,Barenboim2010}\tabularnewline
\hline 
\end{tabular}
\par\end{centering}
\caption{Comparison of algorithms for the problems of MIS, $\left(2\Delta-1\right)$-edge-coloring
and maximal matching.\label{tab:otherProblemsComparisonTable}}
\end{table}

\section{Preliminaries\label{sec:Preliminaries}}

The \emph{k-vertex coloring problem}, sometimes referred to in short
as ``\emph{graph coloring}'', is the problem of assigning each vertex
of an undirected graph a color, such that no two adjacent vertices
share a color. In addition, at most $k$ different colors may be used.
The problem of \emph{k-edge coloring} is the problem of assigning
each edge a color from a set of \emph{k} different colors, such that
each pair of edges which share an endpoint are assigned different
colors. The problem of finding a \emph{maximal independent set} in
a graph is a problem, where given an undirected graph $G=(V,E)$,
one needs to find a subset $I\subseteq V$ such that no edge exists
that connects any two different vertices $v_{1},v_{2}\in I$ and for
each vertex $u\in V\setminus I$, if \emph{u} is added to \emph{I},
then \emph{$I\cup\{u\}$} is no longer an independent set. The problem
of finding a \emph{maximal matching} in a given graph $G=(V,E)$ is
the problem of finding a subset of edges $E'\subseteq E$, such that
for each pair of different edges $e_{1},e_{2}\in E'$, the two edges
have no shared endpoint (vertex), and the addition of any edge $e$
from $E\setminus E'$ to $E'$ will result in the set $E'\cup\{e\}$
no longer being a matching.

We call a vertex that has not yet finished executing a given distributed
algorithm an \emph{active vertex}. Accordingly, we call a vertex that
has finished executing a given distributed algorithm and no longer
takes part in it, an \emph{inactive} \emph{vertex}. For a graph $G=(V,E)$,
we denote $\Delta(G)={\displaystyle \max_{v\in V}\deg(v)}$, where
$\deg(v)$ is the number of edges in $E$ incident on $v$. When the
graph $G$ is clear from context, we simply write $\Delta$. Also,
given a graph $G=(V,E)$ and a subset of vertices $V'\subseteq V$,
we denote by $G(V')$ the subgraph of $G$ induced by $V'$.

Next, we present some definitions regarding edge orientations based
on \cite{Barenboim2013}. An orientation $\mu$ is an assignment of
directions to the edges of a graph $G$, such that every edge $\left\{ u,v\right\} $
of $G$ is directed either towards $u$ or towards $v$. If an orientation
does not contain any consistently oriented cycles (simple cycles in
which each vertex has out-degree 1 and in-degree 1), then we call
it an \emph{acyclic orientation}.

Also, the \emph{out-degree }of an acyclic orientation $\mu$ of $G$
(or, shortly, $\mu$\emph{-out-degree}) is the maximum out-degree
of a vertex in $G$ with respect to $\mu$. In addition, the \emph{length}
of an acyclic orientation $\mu$ of $G$ is the length of the longest
directed path in $G$ with respect to $\mu$. Lastly, for an edge
$e=\left(u,v\right)$ oriented towards $v$ by $\mu$, the vertex
$v$ is called the \emph{parent} of $u$ under $\mu$. Conversely,
The vertex $u$ is called the \emph{child} of $v$ under $\mu$.

\section{Basic Techniques \label{subsec:The-Partition-Algorithm}}

\subsection{Procedure Partition\emph{\label{subsec:Procedure-Partition-Description}}}

In this section we present a basic building block that is used by
many of our algorithms. This is an algorithm devised in \cite{Barenboim2008}
(See also chapter 5 in \cite{Barenboim2013}). The worst-case running
time of the algorithm is $O\left(\log n\right)$. In this section,
however, we demonstrate that its vertex-averaged complexity is significantly
better, specifically, it is $O(1)$. The algorithm is Procedure Partition,
which receives as input an undirected graph $G=(V,E)$, the arboricity
of the graph and a constant $0<\epsilon\le2$ and produces as output
a partition of the graph's vertices into $\ell=\lfloor\frac{2}{\epsilon}\log n\rfloor$
disjoint subsets $H_{1},H_{2},...,H_{\ell}$ such that every vertex
$v\in H_{i}$ has at most $A=(2+\epsilon)\cdot a$ neighbors in the
set $\bigcup_{j=i}^{\ell}H_{j}$. We call each subset $H_{i}$ an
\emph{$H$-set}. Procedure Partition is a subroutine in Procedure
Forest-Decomposition (also presented in \cite{Barenboim2013}), which
partitions the edges of an input graph into $O(a)$ directed forests.
Procedure Forest-Decomposition is in turn used as a subroutine in
additional algorithms presented in \cite{Barenboim2013} for the symmetry-breaking
problems defined in Section \ref{sec:Preliminaries}.

In Procedure Partition, all vertices are active at the start of the
algorithm's execution. Every vertex with at most $A$ neighbors in
the \emph{$i$}-th round of the algorithm's execution joins a subset
of vertices $H_{i}$ and becomes inactive. It is shown in \cite{Barenboim2013}
that each vertex $v\in V$ eventually joins some subset $H_{i}$,
becoming inactive.

It is also shown in \cite{Barenboim2013} that the algorithm has a
worst-case running time of $O(\log_{\frac{2+\epsilon}{2}}n)$ rounds.
Let $n_{i}$ denote the number of active vertices in the input graph
\emph{G} in round \emph{$i$} ($1\le i\le\log_{\frac{2+\epsilon}{2}}n$)
of the algorithm's execution. The following upper bound on $n_{i}$
that is provided in \cite{Barenboim2013} is useful for our analysis.

\begin{lemma}\label{activeVerticesNumberLemma}

\cite{Barenboim2013} For $1\le i\le\ell$, it holds that:
\[
n_{i}\le\left(\frac{2}{2+\epsilon}\right)^{i-1}n
\]
\end{lemma}
\begin{proof}

According to \cite{Barenboim2013}, in each round $i$, for $1\le i\le\log_{\frac{2+\epsilon}{2}}n$,
there are at least $\frac{\epsilon}{2+\epsilon}n_{i}$ vertices with
a degree of at most $A=(2+\epsilon)\cdot a$ and in each round \emph{i}
of the algorithm's execution, all vertices of degree at most $A$
join $H_{i}$ simultaneously, subsequently becoming inactive. Therefore,
for round $i=1,2,...,O(\log n)$, the number of active vertices $n_{i}$
in round $i$ satisfies:
\begin{align*}
n_{i} & \le n_{i-1}-\frac{\epsilon}{2+\epsilon}n_{i-1}\\
 & =(1-\frac{\epsilon}{2+\epsilon})n_{i-1}\\
 & =\frac{2}{2+\epsilon}n_{i-1}\\
 & =\left(\frac{2}{2+\epsilon}\right)^{i-1}n
\end{align*}
as required.

\end{proof}

Also, let $RoundSum(V)$ denote the sum of rounds of all vertices
in $V$ in the execution of Procedure Partition. (For each vertex
we count the number of rounds from the start until it terminates,
and $RoundSum(V)$ is the sum of all these values over all vertices
in $V$.) We observe that if a certain vertex $v\in V$ was active
in rounds $1,2,...,i$ for some $1\le i\le\log_{\frac{2+\epsilon}{2}}n$
then it adds 1 to the value of each of the terms $n_{1},n_{2},...,n_{i}$.
It follows that: 
\begin{equation}
RoundSum(V)=\sum_{i=1}^{\log_{\frac{2+\epsilon}{2}}n}n_{i}\label{eq:RoundSumUpperBoundEquation}
\end{equation}

We present and prove an upper bound on $RoundSum(V)$ within the following
lemma.

\begin{lemma}\label{totalActiveRoundsNumLemma}

It holds that $RoundSum(V)=O(n)$.\end{lemma}
\begin{proof}Applying the result of Lemma \ref{activeVerticesNumberLemma} to Equation
\ref{eq:RoundSumUpperBoundEquation}, we obtain:
\begin{align*}
RoundSum(V) & =\sum_{i=1}^{\log_{\frac{2+\epsilon}{2}}n}n_{i}\\
 & \le\sum_{i=1}^{\log_{\frac{2+\epsilon}{2}}n}\left(\frac{2}{2+\epsilon}\right)^{i-1}n\\
 & =\sum_{i=0}^{\log_{\frac{2+\epsilon}{2}}n-1}\left(\frac{2}{2+\epsilon}\right)^{i}n\\
 & \overset{(1)}{=}O(n)
\end{align*}
where transition (1) follows by computing the sum of the respective
geometric series, while recalling that $0<\epsilon\le2$ is a constant,
and therefore $0<\frac{2}{2+\epsilon}<1$.

\end{proof}

If we take the value of $RoundSum(V)$ given in Lemma \ref{totalActiveRoundsNumLemma}
and divide it by $n$, we obtain the following theorem.

\begin{theorem}\label{PartitionAverageTimePerVertex}

Suppose we are given an input graph $G=(V,E)$, where $\left|V\right|=n$.
Then, the vertex-averaged complexity of Procedure Partition on $G$,
denoted $\bar{T}(G)$, is $\bar{T}(G)=O(1)$.

\end{theorem}

The result of Theorem \ref{PartitionAverageTimePerVertex}, contrasted
with the\emph{ }$O(\log n)$ rounds worst-case running time of Procedure
Partition,\emph{ }implies that by creating in each round a single
$H$-set using Procedure Partition and then performing further operations
on it internally, one can parallelize existing algorithms for the
symmetry-breaking problems defined in Section \ref{sec:Preliminaries},
possibly with some further modifications, obtaining improved vertex-averaged
complexity. This implication serves as the basis for improved algorithms
presented in this paper.

We note that throughout this section, and the rest of the paper, we
assume that the arboricity of the input graph is known to each vertex.
For graphs whose arboricity is unknown, there are standard reductions
from the case of unknown arboricity to the case of known arboricity,
such as Procedure General-Partition in \cite{Barenboim2008}.

\subsection{Analysis of an Algorithm Composed of Procedure Partition and Another
Distributed Algorithm}

Let us define an algorithm $\mathcal{C}$ consisting of $\ell=O(\log n)$
iterations as follows. In each iteration $1\le i\le\ell$, we perform
two steps. The first step is to execute a single round of Procedure
Partition, producing a new $H$-set $H_{i}$. The second step is to
have each vertex $v\in H_{i}$ and these vertices only execute an
auxiliary algorithm $\mathcal{A}$ on the subgraph $G(H_{i})$ induced
by $H_{i}$. Also, let $T_{\mathcal{A}}$ denote the worst-case running
time of algorithm $\mathcal{A}$. From these definitions and Theorem
\ref{PartitionAverageTimePerVertex} the following corollary follows.

\begin{corollary}\label{partitionBasedAlgAverageTimePerVertex}

For an input graph $G=(V,E)$, let \emph{$\mathcal{C}$ }be an algorithm
as above. Then, algorithm $\mathcal{C}$ has a vertex-averaged complexity
of $O(T_{\mathcal{A}})$ rounds.

\end{corollary}

\begin{proof}

Let us denote the sum of the number of communication rounds in which
the vertices of the input graph running $\mathcal{C}$ take part by
$RoundSum(V,T_{\mathcal{\mathcal{A}}})$. Similarly to Lemma \ref{totalActiveRoundsNumLemma},
we observe that:

\begin{align*}
RoundSum(V,T_{\mathcal{A}}) & =\sum_{i=1}^{\log_{\frac{2+\epsilon}{2}}n}n_{i}\cdot O(T_{\mathcal{A}})\\
 & =O(T_{\mathcal{A}})\cdot\sum_{i=1}^{\log_{\frac{2+\epsilon}{2}}n}n_{i}\\
 & =O(T_{\mathcal{A}})\cdot RoundSum(V)
\end{align*}

By Lemma \ref{totalActiveRoundsNumLemma} the last expression satisfies:
\[
O(T_{\mathcal{A}})\cdot RoundSum(V)=O(n\cdot T_{\mathcal{A}})
\]
Therefore, the vertex-averaged complexity of algorithm $\mathcal{C}$
is:
\[
\frac{RoundSum(V,T_{\mathcal{A}})}{n}=\frac{O(n\cdot T_{\mathcal{A}})}{n}=O(T_{\mathcal{A}})
\]
As required.

\end{proof}

It is important to note that though the $\ell$ executions of $\mathcal{A}$
are carried out sequentially on each $H$-set $H_{i}$, and not on
all $H$-sets in parallel, the vertex-averaged complexity of $\mathcal{C}$
is still $O(T_{\mathcal{A}}$).

\section{Deterministic Algorithms}

\subsection{$O(a)$-Forests-Decomposition in $O(1)$ Vertex-Averaged Complexity\label{subsec:Parallelized-Forest-Decomposition}}

In this section we present an average-time analysis of Procedure Forest-Decomposition
devised in \cite{Barenboim2008}. The algorithm receives the same
input as Procedure Partition and returns a partition of the graph's
edges into at most $A=O(a)$ disjoint forests. This partition is also
called an \emph{$O(a)$-forests-decomposition} of the graph's edges.
As for the structure of the algorithm, it first invokes Procedure
Partition with its required input. Then, the algorithm orients each
edge that connects vertices in different $H$-sets towards the vertex
in an $H$-set with a higher index. Each edge connecting vertices
in the same $H$-set is oriented towards the vertex with the higher
ID value. Finally, each vertex $v$, in parallel, labels each of its
outgoing edges arbitrarily, each with a different label from the set
$\left\{ 1,2,...,d_{out}(v)\right\} $, $d_{out}(v)$ being the outgoing
degree of $v$. This produces the desired result in $O\left(\log n\right)$
rounds in the worst case, as shown in \cite{Barenboim2008}.

We improve the vertex-averaged complexity of the procedure as follows.
Immediately upon formation of each subset $H_{i}$ in the $i$-th
iteration of Procedure Partition's main loop, we orient each edge
in the subset towards the endpoint with the higher ID. Simultaneously,
we orient each edge between a vertex $v\in H_{i}$ and a vertex $u\ne v$
not yet belonging to any subset, towards $u$ (since $u$ will join
a subset $H_{j}$ for $j>i$). Afterwards, we label the edges the
same way as in Procedure Forest-Decomposition. We denote the devised
algorithm by\emph{ Procedure Parallelized-Forest-Decomposition}. The
following theorem summarizes the properties of Procedure Parallelized-Forest-Decomposition.

\begin{theorem}\label{parFDTheorem}

Procedure Parallelized-Forest-Decomposition, invoked with an input
graph $G=(V,E)$ with arboricity $a(G)=a$ and a parameter $0<\epsilon\le2$
returns an $O(a)$-forests-decomposition of $E$ with vertex-averaged
complexity $O(1)$.\end{theorem}
\begin{proof}

First, we prove the correctness of the devised procedure. In each
iteration $i=1,2,...,O(\log n)$ of the procedure, each edge with
at least one endpoint in the subset $H_{i}$ created in the current
step is oriented and labeled in the same manner as in Procedure Forest-Decomposition
in \cite{Barenboim2008}. Also, for each step $j=2,3,...,O(\log n)$
no orientation or labeling of any iteration $i<j$ is overridden by
that of the current step $j$. Therefore, the eventual orientation
and labeling of the edges of $E$ is the same as that produced by
Procedure Forest-Decomposition. The orientation and labeling of edges
carried out by Procedure Forest-Decomposition correctly produces a
partition of the edges of $E$ into $O(a)$ forests. Therefore, the
orientation and labeling carried out by Procedure Parallelized-Forest-Decomposition
also correctly produces such a partition of the edges of $E$.

We now analyze the vertex-averaged complexity of Procedure Parallelized-Forest-Decomposition.
In each iteration $i=1,2,...,O(\log n)$ of the procedure's main loop,
at most 3 operations are carried out by each active vertex, all of
which can be carried out in a single round. The first operation is
to decide whether the vertex should join the $H$-set $H_{i}$. If
a vertex decides to join the $H$-set $H_{i}$, then it also performs
the second and third operations, which are to orient and label the
edges incident on it. It follows that the total number of rounds carried
out by the vertices of the input graph in the execution of the devised
procedure is at most $\sum_{i=1}^{\log_{\frac{2+\epsilon}{2}}n}n_{i}$,
similarly to Equation \ref{eq:RoundSumUpperBoundEquation}. According
to the proof of Lemma \ref{totalActiveRoundsNumLemma}, it holds that
$\sum_{i=1}^{\log_{\frac{2+\epsilon}{2}}n}n_{i}=O(n)$ and therefore
the vertex-averaged complexity of the devised procedure is $O(1)$.

\end{proof}

\subsection{$O(a^{2}\log n)$-Vertex-Coloring in $O(1)$ Vertex-Averaged Complexity\label{subsec:O(a^2logn)-Coloring-in-O(1)-Average-Time}}

In this section we employ Procedure Arb-Linial from \cite{Barenboim2008}
(that is based on \cite{Linial1992}), that colors a graph with a
given $O(a)$-forest-decomposition using $O(a^{2}\log n)$ colors.
In general, the procedure first computes an $O(a)$-forest-decomposition
within $O(\log n)$ time, and then computes this coloring within $O(1)$
time. But, given an $O(a)$-forest-decomposition, computing this coloring
from it requires just $1$ round. The procedure used in the second
step of Procedure Arb-Linial, to compute an $O(a^{2}\log n)$-coloring
from a given $O(a)$-forest-decomposition, will be henceforth referred
to as Procedure Arb-Linial-Coloring. Procedure Arb-Linial-Coloring
takes as its input an $O(a)$-forest-decomposition of a graph containing
$n$ vertices, where each vertex in the graph is colored using its
own ID value and produces as output an O($a^{2}$)-coloring of the
input graph. (Actually, we will execute only one step of the procedure
that transforms an $n$-coloring into an $O(a^{2}\log n)$ coloring.
In general, the procedure consists of $O\left(\log^{*}n\right)$ steps.
Each of the $O(\log^{*}n)$ steps transforms a current $p$-coloring
into an $O(a^{2}\log p)$ coloring, for a positive integer $p$. Then,
another single round is executed to reduce a current $O(a^{2}\log a)$-coloring
into an $O(a^{2})$-coloring.)

Our algorithm proceeds as follows. We execute Procedure Parallelized-Forest-Decomposition.
In each of the $\ell=O(\log n)$ iterations of the main loop of the
procedure, we perform two steps. The first step, as before, is to
form an $H$-set $H_{i}$ and to decompose the edges of $G(H_{i})$
into $O(a)$ forests. The second step is to color the vertices of
the current subset $H_{i}$ using $O(a^{2}\log n)$ colors using a
single round of Procedure Arb-Linial-Coloring.

We present and prove the correctness and vertex-averaged complexity
of the devised algorithm in the following theorem.

\begin{theorem}\label{singleRoundArbLinialColoringTheorem} 

For an input graph $G=(V,E)$ with arboricity $a$, the devised algorithm
achieves an $O(a^{2}\cdot\log n)$-coloring in a vertex-averaged complexity
of $O(1)$ rounds.

\end{theorem}

\begin{proof}

First, we prove the coloring achieved is proper. For each $1\le i\le\ell$
the coloring within subset $H_{i}$ is proper, because it was obtained
from invoking a single round of Procedure Arb-Linial-Coloring on the
vertices of $H_{i}$, a technique which correctly produces a proper
$O(a^{2}\log n)$ coloring in this invocation, as proven within Theorem
5.2 and Lemma 5.4 in \cite{Barenboim2013}.

Let $S$ be the set of vertices on which the devised algorithm is
invoked. As explained in \cite{Barenboim2013} in Section 3.10 and
in \cite{Barenboim2008}, in the first round of Procedure Arb-Linial-Coloring,
each vertex $v\in S$ computes and assigns itself a subset $F_{ID(v)}$,
unique to its ID value, from a collection $\mathcal{J}$ of $n$ subsets
of $\{1,2,...,5\cdot\lceil A^{2}\log n\rceil\}$. The assigned subset
$F_{ID(v)}$ satisfies, that for each other $A$ subsets $F_{j_{1},}F_{j_{2}},...,F_{j_{A}}$
of $\mathcal{J}$, there is a value $x\in F_{ID(v)}$ that satisfies
$\{x\}\nsubseteq\cup_{k=1}^{A}F_{j_{k}}$. The existence of such a
collection $\mathcal{J}$ is guaranteed by Lemma 3.21 in \cite{Barenboim2013}.
The vertex $v$ also computes the sets $F_{ID(u_{1})},F_{ID(u_{2})},...,F_{ID(u_{|\Pi(v)|})}$
for each of its neighbors $u_{1},u_{2},...,u_{|\Pi(v)|}\in\Pi(v)$,
where $\Pi(v)$ is the set of parents of $v$, and then chooses a
value $x$ as described as its new color. 

We shall prove by induction that for every $H$-set $H_{i}$ the coloring
achieved is proper for the vertices of the sub-graph induced by $\cup_{j=1}^{i}H_{j}$.
For $i=1$, we have already shown that the coloring of the vertices
of the sub-graph induced by a specific $H$-set is proper. Therefore
the coloring of the vertices of the sub-graph induced by $H_{1}$
is proper. 

Suppose that the coloring achieved is proper for the vertices of the
sub-graph induced by the vertices of $\cup_{j=1}^{k}H_{j}$ for a
certain $1<k<\ell$. For $i=k+1$, let us inspect the $H$-set $H_{k+1}$
formed in the $\left(k+1\right)$-th iteration of the algorithm. We
already know that the coloring achieved of the vertices of the sub-graph
induced by $H_{k+1}$ is proper. Let us look at 3 vertices $u,v,w$
where $u\in\cup_{j=1}^{k}H_{j},v\in H_{k+1},w\in V\setminus\cup H_{j=1}^{k+1}$.
Since within Procedure Arb-Linial each vertex chooses a different
color than each of its parents in the forests-decomposition created
beforehand, the vertex $v$ chose a different color than $w$ (regarding
the ID value of $w$ as its initial color) in the $\left(k+1\right)$-th
iteration of the devised algorithm. Let us denote the color of each
vertex $y\in V$ in the current iteration $k+1$ as $c_{y}$. In the
previous iteration, iteration $k$, the vertex $u$ chose a color
$c_{u}\in F_{ID(u)}$ that did not belong to any of the sets $F_{ID(v_{1})},F_{ID(v_{2})},...,F_{ID(v_{|\Pi(u)|})}$
of any of the parents $v_{1},v_{2},...,v_{\left|\Pi(u)\right|}$ of
$u$. Specifically, this color chosen by $u$ satisfies $c_{u}\notin F_{ID(v)}$.
Therefore, necessarily it holds that $c_{u}\not\ne c_{v}$. It follows
overall that $c_{u}\ne c_{v}$ and $c_{v}\ne c_{w}$. As this holds
for any three vertices $u\in\cup_{j=1}^{k}H_{j},v\in H_{k+1},w\in V\setminus\cup H_{j=1}^{k+1}$,
it follows that the coloring achieved is proper for the vertices of
the sub-graph induced by the vertices of $\cup_{j=1}^{k+1}H_{j}$.
Therefore, for every $H$-set $H_{i}$ the coloring achieved is proper
for the vertices of the sub-graph induced by $\cup_{j=1}^{i}H_{j}$.
This holds in particular for $\cup_{j=1}^{\ell}H_{j}$, which constitutes
the entire input graph, obtaining that the coloring is proper for
the entire input graph. Since each vertex $v\in V$ is ultimately
colored using one of $O(A^{2}\log n)$ colors, the maximum amount
of colors at the end of the execution of this algorithm is $O(A^{2}\log n)=O(a^{2}\log n)$.

We now analyze the vertex-averaged complexity of the devised algorithm.
Since for each $H$-set $H_{i}$, for $1\le i\le\ell$, we execute
a single round of Procedure Arb-Linial-Coloring in the second step
of the devised algorithm, according to Corollary \ref{partitionBasedAlgAverageTimePerVertex},
the vertex-averaged complexity of the devised algorithm is $O(1)$
rounds.

\end{proof}

We note that the best currently-known algorithm requires $\Omega\left(\frac{\log n}{\log a+\log\log n}\right)$
time in the worst case for computing an $O(a^{2}\log n)$ coloring
from scratch. Theorem \ref{singleRoundArbLinialColoringTheorem} implies
the following corollary.

\begin{corollary}\label{singleRoundArbLinialColoringCorollary}

For an input graph $G=(V,E)$ with a constant arboricity $a$, our
algorithm computes an $O(\log n)$-coloring with a vertex-averaged
complexity of $O(1)$.

\end{corollary}

\subsection{$O(a^{2})$-Coloring in $O(\log\log n)$ Vertex-Averaged Complexity\label{subsec:O(a^2)ColoringInO(loglogn)AverageTime}}

In this section we devise an algorithm for $O(a^{2})$-coloring that
consists of two phases. The first phase lasts for $O(\log\log n)$
rounds, while the second phase lasts for $O(\log n)$ rounds. However,
most vertices of the input graph terminate within the first phase,
and so the average running time per vertex of the algorithm is only
$O(\log\log n)$. The algorithm proceeds as follows.

In the first phase, we execute Procedure Parallelized-Forest-Decomposition
for $t=\left\lfloor c'\log\log n\right\rfloor $ iterations of its
main loop, for $c'=\log_{\frac{2+\epsilon}{2}}2$. This invocation
constructs $H$-sets $H_{1},H_{2},...,H_{t}$, and orients and labels
the edges with at least one endpoint in subset $H_{i}$, once $H_{i}$
has formed, for iterations $i=1,2,...,t$. Once $t$ sets have been
formed, we run Procedure Arb-Linial-Coloring for $O(\log^{*}n)$ rounds
on them. This results in $O(a^{2})$ colors, rather than $O(a^{2}\log n)$
when the procedure is executed just for a single round. Then we assign
each vertex in the $t$ first $H$-sets formed so far a color $\left\langle c,1\right\rangle $
where $c$ is the color assigned by Procedure Arb-Linial-Coloring\emph{.}
This completes the first phase of our algorithm.

In the second phase, we continue running Procedure Parallelized-Forest-Decomposition
until every vertex has joined some $H$-set. Then, we run Procedure
Arb-Linial-Coloring for $O(\log^{*}n)$ rounds again on subgraph induced
by the sets $H_{t+1}\cup H_{t+2}\cup...\cup H_{\ell}$. Finally, we
assign $\left\langle c,2\right\rangle $ as the final coloring for
each vertex in the subsets $H_{t+1},H_{t+2},...,H_{\ell}$.

We analyze the devised algorithm in the following two lemmas.

\begin{lemma}\label{O(a^2)ColoringO(loglogn)TimeCorrectnessNumColorsLemma}

The devised algorithm properly colors the input graph using $O(a^{2})$
colors.

\end{lemma}

\begin{proof}

First, we prove that the algorithm produces a proper coloring. Let
$c_{v}=\left\langle c'_{v},i_{v}\right\rangle $ be the color assigned
by the devised algorithm to a vertex $v\in V$ where $i_{v}\in\{1,2\}$.
Procedure Arb-Linial-Coloring was invoked on the sub-graph of the
input graph induced by the $t=\left\lfloor c'\log\log n\right\rfloor $
first $H$-sets $H_{1},H_{2},...,H_{t}$. Also, Procedure Arb-Linial-Coloring
correctly colors the vertices of an input graph with arboricity $a$,
given an $O(a)$-forests-decomposition of the input graph's edges,
using $O(a^{2})$ colors, according to \cite{Barenboim2008}. Therefore,
for the $t=\left\lfloor c'\log\log n\right\rfloor $ first $H$-sets
$H_{1},H_{2},...,H_{t}$, for each two vertices $u,v\in\cup_{j=1}^{t}H_{j},u\ne v$
it holds that $c'_{v}\ne c'_{u}$. Analogously, for each two vertices
$u,v\in\cup_{j=t+1}^{O(\log n)}H_{j},u\ne v$, it also holds that
$c'_{u}\ne c'_{v}$. In addition, for each two vertices $u,v,u\ne v$
for which $u\in\cup_{j=1}^{t}H_{j},v\in\cup_{j=t+1}^{O(\log n)}H_{j}$,
it holds that $i_{u}=1\ne2=i_{v}$. Therefore, for any pair of different
vertices $u,v\in V$ in the input graph $G=(V,E)$, it holds that
$c_{v}\ne c_{u}$, and therefore the coloring produced by the devised
algorithm is proper.

As for the number of colors employed by the devised algorithm, each
vertex is eventually assigned an ordered pair $\left\langle c,i\right\rangle $
as its final color, where $c$ is the color assigned to the vertex
in the invocation of Procedure Arb-Linial-Coloring, which produces
a coloring using $O(a^{2})$ different colors and $i\in\{1,2\}$.
Therefore, the number of colors employed by the devised algorithm
is at most twice the maximum number of colors employed by\emph{ }Procedure
Arb-Linial-Coloring. Therefore, the devised algorithm employs $O(a^{2})$
colors, as required.

\end{proof}

\begin{lemma}\label{O(a^2)ColoringO(loglogn)TimeComplexityLemma}

The vertex-averaged complexity of the devised algorithm is $O(\log\log n)$.

\end{lemma}

\begin{proof}

In the first phase of the devised algorithm, $t=\left\lfloor c'\log\log n\right\rfloor $
$H$-sets are formed, the edges in the graph induced by them subsequently
oriented and labeled and then Procedure Arb-Linial-Coloring\emph{
}invoked on them. The worst-case time complexity of Procedure Arb-Linial-Coloring\emph{
}is $O(\log^{*}n)$ and the orienting and labeling steps occur in
constant time. Therefore, the total number of rounds of all vertices
until this point is: 
\[
O\left(n\cdot\left(\log\log n+\log^{*}n\right)\right)=O(n\log\log n)
\]
In the second phase, we partition the remaining vertices not yet belonging
to any subset to subsets $H_{i}$ for$\left\lfloor c'\log\log n\right\rfloor +1\le i\le\ell$.
According to Lemma \ref{activeVerticesNumberLemma}, the number of
active vertices at this point of the algorithm, that are in $\cup_{i=\left\lfloor c'\log\log n\right\rfloor +1}^{O(\log n)}H_{i}$,
for $0<\epsilon\le2$, is at most:
\begin{align*}
\left(\frac{2}{2+\epsilon}\right)^{\left\lfloor c'\log\log n\right\rfloor }n & \le\left(\frac{2}{2+\epsilon}\right)^{c'\log\log n-1}n\\
 & \le\left(\frac{2+\epsilon}{2}\right)\cdot\left(\frac{2}{2+\epsilon}\right)^{\log_{\frac{2+\epsilon}{2}}2\log\log n}n\\
 & =\left(\frac{2+\epsilon}{2}\right)\cdot\left(\frac{2+\epsilon}{2}\right)^{\log_{\frac{2+\epsilon}{2}}\frac{1}{2}\log\log n}n\\
 & =\left(\frac{2+\epsilon}{2}\right)\cdot\left(\frac{1}{2}\right)^{\log\log n}n\\
 & =O\left(\frac{n}{\log n}\right)
\end{align*}
Each of these vertices performs $O(\log n+\log^{*}n)=O(\log n)$ rounds
of computation. It follows that the vertex-averaged complexity of
the devised algorithm is:
\begin{align*}
\frac{O(n\log\log n+\frac{n}{\log n}\cdot\log n)}{n} & =\frac{O(n\log\log n)}{n}=O(\log\log n)
\end{align*}
as required.

\end{proof}

Lemmas \ref{O(a^2)ColoringO(loglogn)TimeCorrectnessNumColorsLemma}, \ref{O(a^2)ColoringO(loglogn)TimeComplexityLemma} 
imply the following theorem.

\begin{theorem}\label{O(a^2)ColoringO(loglogn)TimeTheorem}

The devised algorithm properly colors an input graph using $O(a^{2})$
colors with a vertex-averaged complexity of $O(\log\log n)$ rounds.

\end{theorem}

\subsection{$O\left(a\right)$-Vertex-Coloring in $O\left(a\log\log n\right)$
Vertex-Averaged Complexity\label{subsec:O(a)VertexColoringSectionWithImprovedVAC}}

In this section we devise an algorithm that colors the vertices of
an input graph using $O\left(a\right)$ colors with a vertex-averaged
complexity of $O\left(a\log\log n\right)$ rounds.

The algorithm consists of two phases. The first phase consists of
$t=\left\lfloor \log\log n\right\rfloor $ rounds and proceeds as
follows. To begin with, we execute Procedure Partition. In each iteration
of Procedure Partition, once an $H$-set $H_{i}$ is formed, we do
the following. First, we color the vertices of $H_{i}$ using the
$\left(\Delta+1\right)$-coloring algorithm of \cite{Barenboim2009}.
Next, we compute an acyclic orientation of the edges of the sub-graph
$G\left(H_{i}\right)$, by orienting each edge $e=\left\{ u,v\right\} $
in the set of edges of $G\left(H_{i}\right)$ towards the vertex with
a higher color value. Subsequently, we orient each edge $e=\left\{ u,v\right\} $
connecting a vertex $u\in H_{i}$ with a vertex $v\in H_{j}$ for
$j>i$, towards $v$. Once the the first phase has completed, we recolor
the vertices of $\cup_{i=1}^{t}H_{i}$ using the following method,
used in chapter 5 of \cite{Barenboim2013}. We start the recoloring
at the vertices of $H_{t}$. Each vertex $v$ which has not yet assigned
itself a color, within the recoloring stage, first waits for all of
its parents, with respect to the orientation created in the previous
steps, to first choose a color, and then chooses a new color $c_{v}$
for itself from the palette $\left\{ 1,2,...,A+1\right\} $. Then,
the vertex $v$ assigns itself the color $c'_{v}=\left\langle c_{v},1\right\rangle $.

In the second phase, we perform the same operations as in the first
phase, only that in the recoloring stage, we start the recoloring
at the $H$-set $H_{\log\log n}$ and stop it once we have completed
the recoloring of the vertices of $H_{t+1}$. Also, once a vertex
$v$ chooses for itself a new color $c_{v}\in\left\{ 1,2,...,A+1\right\} $
different than that of each its parents, it then assigns itself the
color $c'_{v}=\left\langle c_{v},2\right\rangle $. This completes
the description of the devised algorithm.

\begin{lemma}\label{O(a)ColoringWithImprovedVACCorrectnessLemma}

The devised algorithm properly colors an input graph's vertices using
$O\left(a\right)$ colors.

\end{lemma}

\begin{proof}

We first prove the devised algorithm computes a proper vertex-coloring
of the input graph using $O\left(a\right)$ colors. We first observe,
that according to Section \ref{subsec:Procedure-Partition-Description},
the maximum degree of each sub-graph $G\left(H_{i}\right)$, induced
by an $H$-set $H_{i}$, is $O\left(a\right)$. Therefore, each invocation
of the $\left(\Delta+1\right)$-coloring algorithm of \cite{Barenboim2009}
on an $H$-set $H_{i}$ colors it using $O\left(a\right)$ colors.

Next, we prove, that in the step of orienting the edges $e=\left\{ u,v\right\} $
where $u,v\in H_{i}$, for some $H$-set $H_{i}$, we have computed
an acyclic orientation of the edges of $G\left(H_{i}\right)$. The
proof we present is based on Property 3.4 in \cite{Barenboim2013}.
We remind that in this step of the algorithm, we oriented each edge
$e=\left\{ u,v\right\} $ towards the vertex with the larger color
value. Therefore, following this operation, each directed path in
$G\left(H_{i}\right)$ consists of vertices whose colors appear in
a strictly ascending order. Therefore, the produced orientation of
the edges of $G\left(H_{i}\right)$ is acyclic.

In addition, the orientation created by the algorithm for edges $e=\left\{ u,v\right\} $
for $u\in H_{i},v\in H_{j}$, is also acyclic, according to Section
\ref{subsec:Parallelized-Forest-Decomposition}. Moreover, each vertex
$v\in H_{i}$ has at most $A$ neighbors in $\cup_{j=i}^{\log n}H_{j}$,
according to Section \ref{subsec:Procedure-Partition-Description}.

Therefore, in the recoloring step of the devised algorithm, each vertex
properly re-assigns itself once a color different than that of each
of its parents, with respect to the orientation created by the devised
algorithm, if such a color exists. Also, since each vertex $v\in H_{i}$
has at most $A$ neighbors in $\cup_{j=i}^{\log n}H_{j}$, such a
color necessarily exists. Therefore, the devised algorithm properly
colors each of the sub-graphs $G\left(\cup_{i=1}^{t}H_{i}\right)$,
$G\left(\cup_{i=t+1}^{\log\log n}H_{i}\right)$ of $G$. In addition,
the color assigned to each vertex $v$, where $v\in G\left(\cup_{i=1}^{t}H_{i}\right)$,
or $v\in G\left(\cup_{i=t+1}^{\log\log n}H_{i}\right)$, is of the
form $\left\langle c_{v},i\right\rangle $. The color $c_{v}$ is
the color $v$ assigned itself in the recoloring step of the devised
algorithm from the palette $\left\{ 1,2,...,A+1\right\} $, and for
each two vertices $u,v$, with colors $\left\langle c_{u},i\right\rangle ,\left\langle c_{v},j\right\rangle $,
it holds that $i\ne j$. Therefore, the devised algorithm properly
colors the input graph using at most $O\left(a\right)$ colors, as
required.

\end{proof}

\begin{lemma}\label{O(a)VertexColoringImprovedVACTimeComplexityLemma}

The devised algorithm's vertex averaged complexity is $O\left(a\log\log n\right)$
rounds.

\end{lemma}

\begin{proof}

In each of the two phases, executing all rounds of Procedure Partition
requires $O\left(n\right)$ rounds in the worst case, according to
Section \ref{subsec:Procedure-Partition-Description}. Each invocation
of the $\left(\Delta+1\right)$-coloring algorithm of \cite{Barenboim2009}
requires $O\left(\Delta+\log^{*}n\right)$ rounds, where $\Delta$
in this context is the maximum degree in a subgraph induced by a certain
$H$-set $H_{i}$. Therefore, each invocation of the $\left(\Delta+1\right)$-coloring
algorithm of \cite{Barenboim2009} requires $O\left(a+\log^{*}n\right)$
rounds. Each of the steps of orienting edges requires a constant number
of rounds.

In the first phase, the recoloring step requires a number of rounds
that is asymptotically equal to the maximum length of the orientation
created by the previous steps that included the orienting of edges.
This orientation is acyclic, as shown in the proof of Lemma \ref{O(a)ColoringWithImprovedVACCorrectnessLemma}.
Therefore, the length of the orientation within a certain $H$-set
$H_{i}$ is $O(a)$. In addition, since $t=O\left(\log\log n\right)$
are created and handled in the first phase, the length of the orientation
of the edges of $G\left(\cup_{i=1}^{t}H_{i}\right)$ is $O\left(a\log\log n\right)$.
Therefore, the worst-case number of rounds required by the recoloring
step in the first phase is $O\left(a\log\log n\right)$. 

In the second phase, following a similar analysis to that of the worst-case
running time of the first phase of the devised algorithm, the worst-case
running time of the second phase is $O\left(a\log n\right)$.

According to Lemma \ref{activeVerticesNumberLemma}, in iterations
$t+1\le i\le O\left(\log n\right)$ of the devised algorithm, that
is, in the second phase of the algorithm, the number of active vertices
is $O\left(\frac{n}{2^{\log\log n}}\right)=O\left(\frac{n}{\log n}\right)$.
In the first phase, trivially at most $n$ vertices are active. Therefore,
the vertex-averaged complexity of the devised algorithm is:
\[
O\left(\frac{na\log\log n+\frac{n}{\log n}\cdot a\log n}{n}\right)=a\log\log n
\]

as required.

\end{proof}

We summarize the properties of the devised algorithm in the following
theorem.

\begin{theorem}\label{O(ak)DeterministicVertexColoringTheorem}

The devised algorithm properly colors the vertices of an input graph
using $O\left(a\right)$ colors, with a vertex-averaged complexity
of $O\left(a\log\log n\right)$ rounds.

\end{theorem}

\subsection{Segmentation: General $O\left(k\alpha\right)$-Vertex-Coloring Scheme\label{subsec:O(k * alpha)VertexColoringInImprovedVertexAveragedComplexity}}

In this section we devise a general scheme, which we call \emph{segmentation,
}that we will use to generalize the algorithms of Sections \ref{subsec:O(a^2)ColoringInO(loglogn)AverageTime},
\ref{subsec:O(a)VertexColoringSectionWithImprovedVAC}. 

First, we present an intuition regarding the devised scheme. Roughly
speaking, we divide the input graph's vertices into $k$ subsets.
The value $k$ belongs to the set $\left\{ 2,3,...,\rho\left(n\right)\right\} $,
where $\rho\left(n\right)$ is the largest integer, such that $\log^{\left(\rho\left(n\right)-1\right)}n\ge\log^{*}n$
. We chose $\rho\left(n\right)$ as the maximum possible value $k$
can take, rather than $\log^{*}n$, for instance, to help obtain useful
results, as will become clearer in a later analysis in this section.
For each $i=1,2,...,k$ the $i$-th subset consists of roughly $\log^{(i)}n$
$H$-sets. Each of these $k$ subsets will be henceforth denoted as
a \emph{segment}. The first segment to be created during the execution
of the scheme is segment $k$, consisting of $\log^{\left(k\right)}n$
subsets, then segment $k-1$, consisting of $\log^{\left(k-1\right)}n$
subsets and so on, until all vertices of the graph belong to some
$H$-set, and also to some segment.

On each of these segments, we execute a slightly modified version
of each of the algorithms of Sections \ref{subsec:O(a^2)ColoringInO(loglogn)AverageTime},
\ref{subsec:O(a)VertexColoringSectionWithImprovedVAC}. The main modification
is that instead of executing the algorithm of each of these mentioned
sections over only 2 phases, we use $k$ phases. We count the phases,
from this point onward, from $k$ down to $1$. Each phase $i$ is
executed on segment $i$. That is, we execute phase $k$ on segment
$k$, then phase $k-1$ on segment $k-1$, and so on. In each phase,
a different palette of colors is used to internally color each segment,
which results in a total number of colors of $O\left(k\cdot\alpha\right)$,
where $\alpha=O\left(a^{2}\right)$, in the case of the algorithm
of Section \ref{subsec:O(a^2)ColoringInO(loglogn)AverageTime}, and
$\alpha=O\left(a\right)$, in the case of the algorithm of Section
\ref{subsec:O(a)VertexColoringSectionWithImprovedVAC}.

We now explain the effect of the use of the scheme on the vertex-averaged
complexity of the algorithms obtained from using it. Based on the
exponential decay in the number of active vertices throughout the
execution of Procedure Partition, as explained in Section \ref{subsec:Procedure-Partition-Description},
in each iteration $2\le i\le k$, the number of active vertices is
$O\left(\frac{n}{\log^{(i)}n}\right)$ . Also, following an analysis
of the worst-case running time of each phase, similarly to the proofs
of Lemmas \ref{O(a^2)ColoringO(loglogn)TimeComplexityLemma}, \ref{O(a)VertexColoringImprovedVACTimeComplexityLemma},
the worst-case time complexity of phase $i$, is $O\left(\log^{\left(i\right)}n\right)$,
in the case of the algorithm of Section \ref{subsec:O(a^2)ColoringInO(loglogn)AverageTime},
or $O\left(a\log^{\left(i\right)}n\right)$, in the case of the algorithm
of Section \ref{subsec:O(a)VertexColoringSectionWithImprovedVAC}.

Figure \ref{fig:segmentationFigure} provides an illustration of an
execution of the general scheme, when applied to the algorithm of
Section \ref{subsec:O(a)VertexColoringSectionWithImprovedVAC}.

\begin{figure}
\includegraphics[width=1\columnwidth]{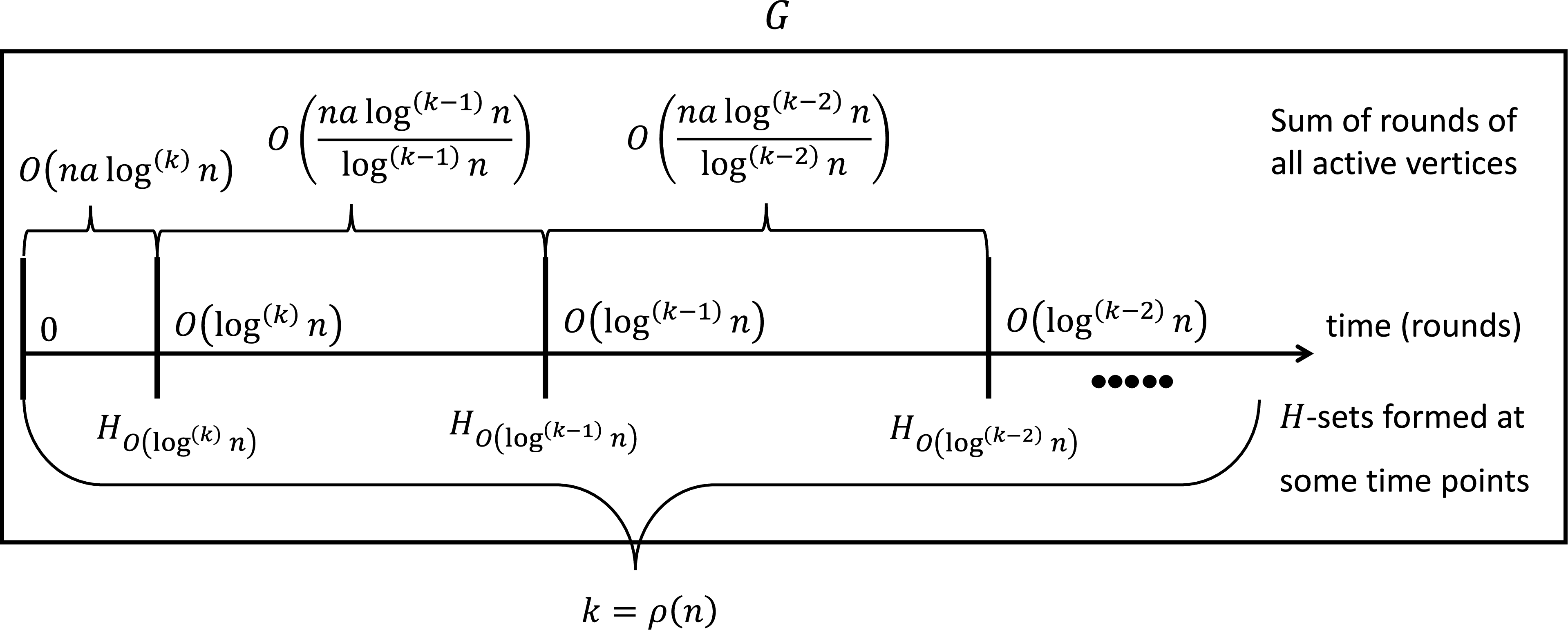}

\caption{An illustration of an execution of the general scheme for $k=\rho\left(n\right)$.\label{fig:segmentationFigure}}
\end{figure}

We now present a formal description of the general scheme. For the
invocation of this general scheme, we assume we are given algorithms
$\mathcal{A}$,$\mathcal{B}$ and $\mathcal{C}$, whose properties
we shall describe shortly. The scheme devised in this section generally
proceeds as follows. In a loop from $i=k$ down to $1$, in each iteration
$i$, we first produce up to $O\left(\log^{\left(i\right)}n\right)$
subsets $H_{j}$ using Procedure Partition, which we also call a segment,
as mentioned above. Upon the formation of each subset $H_{j}$, we
invoke algorithm $\mathcal{A}$, then algorithm $\mathcal{B}$, and
then orient all edges which have one endpoint in $H_{j}$, and the
other endpoint outside any $H$-set, towards the latter endpoint.
In parallel to the execution of algorithms $\mathcal{A},\mathcal{B}$,
and the orientation of edges on the last $H$-set $H_{j}$ formed
by the invocation of Procedure Partition, we do the following. We
continue to execute Procedure Partition to form new $H$-sets, and
run algorithms $\mathcal{A}$ and $\mathcal{B}$, and the orientation
of edges, on these new $H$-sets, until $O\left(\log^{\left(i\right)}n\right)$
$H$-sets have been formed in the current iteration $i$, and algorithms
$\mathcal{A},\mathcal{B}$, and the orientation of edges have finished
execution on each $H$-set formed in the current iteration $i$. 

No general constraints are imposed on the selection of an algorithm
$\mathcal{A}$. The purpose of algorithm $\mathcal{A}$, if it is
necessary to invoke such an algorithm, is to produce output required
by algorithms $\mathcal{B}$, $\mathcal{C}$ when they are later executed.
Algorithm $\mathcal{B}$ takes as input the sub-graph $G\left(H_{j}\right)$,
induced on an $H$-set $H_{j}$, and produces as output an acyclic
orientation of the edges of $G\left(H_{j}\right)$ with length $\lambda$
and out-degree $O\left(a\right)$, for an appropriate parameter $\lambda$.

Suppose that in the current iteration $i$ of the devised scheme,
the $H$-sets $H_{i_{1}},H_{i_{2}},...,H_{i_{l}}$ have been formed.
Algorithm $\mathcal{C}$ begins execution when algorithms $\mathcal{A},\mathcal{B}$,
and the step of the orientation of edges finish their execution on
$H_{i_{1}},H_{i_{2}},...,H_{i_{l}}$. Let us denote $\mathcal{H}=\left\{ H_{i_{1}},H_{i_{2}},...,H_{i_{l}}\right\} $.
Algorithm $\mathcal{C}$ takes as input the sub-graph $G\left(\cup_{H\in\mathcal{H}}H\right)$,
together with the orientation $\mu$ computed by algorithm $\mathcal{B}$
on the sub-graph $G\left(\cup_{H\in\mathcal{H}}H\right)$, where the
length of $\mu$ is $O\left(\lambda\log^{\left(i\right)}n\right)$
and its out-degree is $O\left(a\right)$. Then, algorithm $\mathcal{C}$
colors $G\left(\cup_{H\in\mathcal{H}}H\right)$ using $O\left(\alpha\right)$
colors. 

In addition, we denote the worst-case time complexity of algorithms
$\mathcal{A}$, $\mathcal{B}$ by $T_{\mathcal{A}},T_{\mathcal{B}}$,
respectively. We also denote the worst-case time complexity of algorithm
$\mathcal{C}$ on a subgraph of the form $G\left(\cup_{H\in\mathcal{H}}H\right)$
by $T_{\mathcal{C},l}$. We now formally detail the steps of the devised
scheme. The devised scheme proceeds as follows:
\begin{enumerate}
\item For $i=k$ down to $1$:
\begin{enumerate}
\item Let $c=\frac{2}{\epsilon}$.
\item For $j=1$ to $c\cdot\log^{\left(i\right)}n$:
\begin{enumerate}
\item Perform a round of Procedure Partition.
\item Upon the formation of the $H$-set created in the current iteration
$j$ of the loop of step 1(b), invoke algorithm $\mathcal{A}$ on
it.
\item Once algorithm $\mathcal{A}$ completes execution on the $H$-set
computed in step 1(b)i, execute algorithm $\mathcal{B}$ on it.
\item Suppose the $H$-set formed in the current iteration of the loop of
step 1(b) is $H_{l}$. For each vertex $u\in H_{l}$, for each edge
$e=\left\{ u,v\right\} \in E$, such that $v$ does not yet belong
to any $H$-set, orient $e$ towards $v$.
\item In parallel to the execution of steps 1(b)ii, 1(b)iii, 1(b)iv on the
$H$-set formed in the current iteration $j$ of the loop of step
1(b), if $j<c\log^{\left(i\right)}n$, continue to the next iteration
$j+1$.
\end{enumerate}
\item Invoke algorithm $\mathcal{C}$ on the sub-graph $G\left(H\right)$
induced by the set of vertices $H\subseteq V$, which consists of
all vertices which have joined an $H$-set in the current iteration
$i$ of the devised scheme, using the palette $\left\{ \left(i-1\right)\alpha,\left(i-1\right)\alpha+1,...,i\alpha-1\right\} $.
\end{enumerate}
\end{enumerate}
The following lemma describes the correctness of the scheme and the
number of colors the scheme employs. The main idea is that each of
the $k$ segments is properly colored using a palette with $O\left(\alpha\right)$
unique colors, resulting in an overall $O(k\alpha)$-coloring of the
input graph $G$.

\begin{lemma}\label{O(k * alpha)VertexColoringCorrectnessLemma}

The devised scheme computes a proper coloring of an input graph $G=\left(V,E\right)$
using $O\left(k\alpha\right)$ colors.

\end{lemma}

\begin{proof}

Denote by $G_{i}$ the subgraph of $G$ induced by the vertices of
segment $i$. We first prove that in each invocation of algorithm
$\mathcal{C}$ it receives an input of the form it requires. Suppose
that in the current invocation of algorithm $\mathcal{C}$ in iteration
$i$ of the loop of step 1 of the devised algorithm, it is invoked
on a sub-graph $G_{i}$. The execution of algorithm $\mathcal{B}$
before algorithm $\mathcal{C}$ guarantees that the vertices of $G_{i}$
are partitioned into $H$-sets, as required. By the properties of
algorithm $\mathcal{B}$, for each $H$-set $H_{j}$ in $G_{i}$,
the orientation of the edges of $G\left(H_{j}\right)$, produced by
algorithm $\mathcal{B}$, is acyclic, and has length at most $\lambda$.
Also, the orientation carried out in step 1(b)iv of the devised algorithm
guarantees that the orientation of edges connecting vertices in different
$H$-sets in $G_{i}$ is also acyclic. It follows that the orientation
on the edges of $G_{i}$ is acyclic. Also, since there are $O\left(\log^{\left(i\right)}n\right)$
$H$-sets in $G_{i}$, the orientation of the edges of $G_{i}$ has
length $O\left(\lambda\log^{\left(i\right)}n\right)$, as required.

We now prove that the out-degree of the orientation of the edges of
$G_{i}$ is $O\left(a\right)$, as required. By the properties of
the output of algorithm $\mathcal{B}$, the out-degree of the orientation
on each sub-graph $G\left(H_{j}\right)$, for an $H$-set $H_{j}$
in $G_{i}$, is $O\left(a\right)$. Also, according to Section \ref{subsec:Procedure-Partition-Description},
each vertex $v\in H_{j}$ has $O\left(a\right)$ neighbors in $H$-sets
$H_{l}$, for $l\ge j$. Therefore, overall, the out-degree of the
orientation of the edges of $G_{i}$ is $O\left(a\right)$, as required.

We now prove that the devised algorithm computes a proper $O\left(k\alpha\right)$
coloring of the input graph $G$. For each sub-graph $G_{i}$ of $G$,
by definition of algorithm $\mathcal{C}$, the sub-graph $G_{i}$
is properly colored. For two different vertices $u\in G_{i},v\in G_{j}$,
for $i\ne j$, the vertex $u$ was colored using algorithm $\mathcal{C}$
using the palette $\left\{ \left(i-1\right)\alpha,\left(i-1\right)\alpha+1,...,i\alpha-1\right\} $,
and the vertex $v$ was colored using the palette $\left\{ \left(j-1\right)\alpha,\left(j-1\right)\alpha+1,...,j\alpha-1\right\} $.
Therefore, each of the two vertices $u,v$ necessarily received a
different color.

Therefore, for each pair of different vertices $u,v\in V$, $u,v$
have received a different color. Therefore, the input graph $G$ is
properly colored, as required.

We now analyze the number of colors employed by the devised scheme.
In each iteration $1\le i\le k$ of step 1 of the devised scheme,
$\alpha$ colors are used by algorithm $\mathcal{C}$ to color the
vertices of $G_{i}$. Since these mentioned color palettes are all
disjoint, the total number of colors used by the devised scheme is
$O\left(k\alpha\right)$, as required.

\end{proof}\label{O(k * alpha)VertexColoringVertexAveragedTimeLemma}

The following lemma describes the vertex-averaged complexity of the
devised scheme.

\begin{lemma}

The vertex-averaged complexity of the devised scheme is:
\[
O\left(\log^{\left(k\right)}n+T_{\mathcal{A}}+T_{\mathcal{B}}+\sum_{i=1}^{k-1}\frac{T_{\mathcal{C},i}}{\log^{\left(i\right)}n}+T_{\mathcal{C},k}\right)
\]

\end{lemma}

\begin{proof}

We first analyze the total number of communication rounds in the first
iteration $i=k$ of step 1 of the devised scheme. The total number
of communication rounds required for the execution of $O\left(\log^{\left(i\right)}n\right)$
rounds of Procedure Partition, and for the orientation of edges in
step 1(b)iv is $O\left(n\log^{\left(k\right)}n\right)$ rounds. The
total number of communication rounds required for the execution of
algorithms $\mathcal{A}$,$\mathcal{B}$ and $\mathcal{C}$ is $O\left(n\left(T_{\mathcal{A}}+T_{\mathcal{B}}+T_{\mathcal{C},k}\right)\right)$.
Therefore, the total number of communication rounds carried out in
iteration $i=k$ of step 1 of the devised algorithm is $O\left(n\left(\log^{\left(k\right)}n+T_{\mathcal{A}}+T_{\mathcal{B}}+T_{\mathcal{C},k}\right)\right)$.

We now analyze the total number of communication rounds in iteration
$i<k$ of step 1 of the devised scheme. According to Section \ref{subsec:Procedure-Partition-Description},
in each round $j$ of Procedure Partition, where $n_{j}$ vertices
are active, at least a constant fraction of these $n_{j}$ vertices
terminate. We run Procedure Partition in each iteration $1\le i\le k$
of the scheme for $O\left(\log^{\left(i\right)}n\right)$ rounds.
Therefore, at the end of an iteration $1\le l\le k$, only roughly
$O\left(\frac{n}{2^{\log^{\left(l\right)}n}}\right)$ vertices continue
to run subsequent iterations (if such subsequent iterations remain).
Therefore, at the beginning of iteration $i=l-1$, for $1\le i<k$,
the number of active vertices executing iteration $i$ is only:
\[
O\left(\frac{n}{2^{\log^{\left(l\right)}n}}\right)=O\left(\frac{n}{2^{\log^{\left(i+1\right)}n}}\right)=O\left(\frac{n}{\log^{\left(i\right)}n}\right)
\]
 Therefore, following a similar analysis to that of the total number
of communication rounds carried out in iteration $k$ of step 1 of
the devised scheme, the total number of communication rounds carried
out in iteration $i$ of step 1, for $i<k$, is:
\begin{align*}
O\left(\frac{n}{\log^{\left(i\right)}n}\left(\log^{\left(i\right)}n+T_{\mathcal{A}}+T_{\mathcal{B}}+T_{\mathcal{C},i}\right)\right) & =O\left(n+\frac{n}{\log^{\left(i\right)}n}\left(T_{\mathcal{A}}+T_{\mathcal{B}}+T_{\mathcal{C},i}\right)\right)
\end{align*}

Therefore the total number of communication rounds carried out by
the vertices of the input graph in the execution of the devised scheme
is:
\begin{align*}
O\left(n\log^{\left(k\right)}n+nT_{\mathcal{A}}+nT_{\mathcal{B}}+nT_{\mathcal{C},k}+nk+\sum_{i=1}^{k-1}\frac{n}{\log^{\left(i\right)}n}\left(T_{\mathcal{A}}+T_{\mathcal{B}}+T_{\mathcal{C},i}\right)\right) & =\\
O\left(n\log^{\left(k\right)}n+nT_{\mathcal{A}}+nT_{\mathcal{B}}+nT_{\mathcal{C},k}+nk+\frac{kn}{\log^{\left(k-1\right)}n}\left(T_{\mathcal{A}}+T_{\mathcal{B}}\right)+\sum_{i=1}^{k-1}\frac{nT_{\mathcal{C},i}}{\log^{\left(i\right)}n}\right) & =\\
O\left(n\log^{\left(k\right)}n+nk+n\left(1+\frac{k}{\log^{\left(k-1\right)}n}\right)\left(T_{\mathcal{A}}+T_{\mathcal{B}}\right)+\sum_{i=1}^{k-1}\frac{nT_{\mathcal{C},i}}{\log^{\left(i\right)}n}+nT_{\mathcal{C},k}\right)
\end{align*}
In addition, we chose $2\le k\le\rho\left(n\right)$. By definition
of $\rho\left(n\right)$, it holds that $k\le\log^{*}n$, as otherwise
it would hold that $\log^{\left(\rho\left(n\right)-1\right)}n=O(1)$,
and specifically, that $\log^{\left(\rho\left(n\right)-1\right)}n<\log^{*}n$,
contradicting the definition of $\rho\left(n\right)$. Also, by definition
of $\rho\left(n\right)$, it holds that $\log^{\left(k-1\right)}n\ge\log^{*}n$.
Therefore, the last expression given for the total number of communication
rounds carried out by the vertices of the input graph in the execution
of the devised scheme is equal to:
\[
n\cdot O\left(\log^{\left(k\right)}n+T_{\mathcal{A}}+T_{\mathcal{B}}+\sum_{i=1}^{k-1}\frac{T_{\mathcal{C},i}}{\log^{\left(i\right)}n}+T_{\mathcal{C},k}\right)
\]
Therefore, the vertex-averaged complexity of the devised scheme is:
\[
O\left(\log^{\left(k\right)}n+T_{\mathcal{A}}+T_{\mathcal{B}}+\sum_{i=1}^{k-1}\frac{T_{\mathcal{C},i}}{\log^{\left(i\right)}n}+T_{\mathcal{C},k}\right)
\]
as required.

\end{proof}

The properties of the scheme are summarized in the following theorem.

\begin{theorem}\label{O(k * alpha)VertexColoringSchemeTheorem}

The devised scheme properly colors an input graph using $O\left(k\alpha\right)$
colors, with a vertex-averaged complexity of $O\left(\log^{\left(k\right)}n+T_{\mathcal{A}}+T_{\mathcal{B}}+\sum_{i=1}^{k-1}\frac{T_{\mathcal{C},i}}{\log^{\left(i\right)}n}+T_{\mathcal{C},k}\right)$
rounds.

\end{theorem}

In the following two sections, we describe how we apply the scheme
devised in this section to the algorithms of Sections \ref{subsec:O(a^2)ColoringInO(loglogn)AverageTime},
\ref{subsec:O(a)VertexColoringSectionWithImprovedVAC}, to obtain
a generalized version of each algorithm, as well as useful special
cases.

\subsection{$O(ka^{2})$-Vertex-Coloring in $O\left(\log^{\left(k\right)}n\right)$
Vertex-Averaged Complexity\label{subsec:O(ka^2)Coloring}}

In this section we generalize the algorithm of Section \ref{subsec:O(a^2)ColoringInO(loglogn)AverageTime},
using the general scheme of Section \ref{subsec:O(ka^2)Coloring}.
To this end, we devise an algorithm that colors the vertices of an
input graph using $O\left(ka^{2}\right)$ colors within a vertex-averaged
complexity of $O\left(\log^{\left(k\right)}n\right)$ rounds, for
$2\le k\le\rho\left(n\right)$. For $k=\rho\left(n\right)$, we obtain
an $O\left(a^{2}\cdot\log^{*}n\right)$-vertex-coloring with a vertex
averaged complexity of $O\left(\log^{*}n\right)$ rounds.

As algorithm $\mathcal{A}$ of the scheme of Section \ref{subsec:O(k * alpha)VertexColoringInImprovedVertexAveragedComplexity},
we use a ``null'' algorithm that does nothing. As algorithm $\mathcal{B}$
of the scheme, we invoke Procedure Parallelized-Forest-Decomposition.
As algorithm $\mathcal{C}$ of the same scheme, we use Procedure Arb-Linial-Coloring,
which was defined in Section \ref{subsec:O(a^2logn)-Coloring-in-O(1)-Average-Time},
and is a step in the $O\left(a^{2}\right)$-vertex-coloring algorithm
of \cite{Barenboim2008}.

\begin{theorem} \label{O(ka^2)VertexColoringTheorem}

The devised algorithm computes an $O\left(ka^{2}\right)$-vertex-coloring
of an input graph with a vertex-averaged complexity of $O\left(\log^{\left(k\right)}n\right)$
rounds.

\end{theorem}

\begin{proof}

We first prove the devised algorithm computes a proper coloring of
the vertices of the input graph. To do this, we first prove that the
chosen algorithm $\mathcal{B}$ satisfies the properties required
by the scheme of Section \ref{subsec:O(k * alpha)VertexColoringInImprovedVertexAveragedComplexity}.
The value of $\lambda$ in this case is irrelevant. The reason is
that the chosen algorithm $\mathcal{C}$ only requires that the edges
of the sub-graph it is invoked on are decomposed into $O\left(a\right)$
forests, which is achieved by the chosen algorithm $\mathcal{B}$,
which is in fact an invocation of Procedure Parallelized-Forest-Decomposition.
However, an upper bound on the length of the orientation of the sub-graph
on which algorithm $\mathcal{C}$ is invoked is not needed. Also,
since algorithm $\mathcal{B}$ consists of an invocation of Procedure
Parallelized-Forest-Decomposition, then according to Section \ref{subsec:Parallelized-Forest-Decomposition},
it decomposes the edges of the sub-graph $G\left(H_{j}\right)$ on
which it is invoked, for an $H$-set $H_{j}$, into $O\left(a\right)$
directed forests. Therefore, the out-degree of the orientation of
the edges of $G\left(H_{j}\right)$ produced by algorithm $\mathcal{B}$
is $O\left(a\right)$, as required.

In addition, we observe, that by the known correctness of the $O\left(a^{2}\right)$-vertex-coloring
algorithm of \cite{Barenboim2008}, algorithm $\mathcal{C}$ properly
colors each sub-graph $G_{i}$ on which it is invoked. Therefore,
according to Theorem \ref{O(k * alpha)VertexColoringSchemeTheorem},
the devised algorithm properly colors the input graph, as required.

As for the number of colors employed by the devised algorithm, according
to Theorem \ref{O(k * alpha)VertexColoringSchemeTheorem}, the number
of colors the devised algorithm employs is $O\left(ka^{2}\right)$,
as required.

We now analyze the vertex-averaged complexity of the devised algorithm.
The chosen algorithm $\mathcal{A}$ does nothing and therefore requires
$0$ rounds. The algorithm chosen as algorithm $\mathcal{B}$ invokes
Procedure Parallelized-Forest-Decomposition on an $H$-set $H_{j}$.
When invoked on an entire input graph, the worst-case time complexity
of Procedure Parallelized-Forest-Decomposition is $O\left(\log n\right)$,
the same as the worst-case time complexity of Procedure Forests-Decomposition
of \cite{Barenboim2008}. However, when only the orientation of edges
is carried out by the algorithm on a sub-graph $G\left(H_{j}\right)$
of $G$, for an $H$-set $H_{j}$, executing algorithm $\mathcal{B}$
requires $O\left(1\right)$ rounds. Therefore, the worst-case time
complexity of algorithm $\mathcal{B}$ is $O\left(1\right)$.

In addition, the worst-case time complexity of algorithm $\mathcal{C}$,
according to \cite{Barenboim2008}, is $O\left(\log^{*}n\right)$
rounds. Therefore, the vertex-averaged complexity of the devised algorithm
is:
\begin{align*}
O\left(\log^{\left(k\right)}n+\sum_{i=1}^{k-1}\frac{n\log^{\left(i\right)}n}{\log^{\left(i\right)}n}\right) & =O\left(\log^{\left(k\right)}n\right)
\end{align*}
as required.

\end{proof}

For $k=\rho\left(n\right)$, we obtain the following corollary.

\begin{corollary}\label{O(a^2*log*n)VertexColoringAlgorithmCorollary}

For $k=\rho\left(n\right)$, the devised algorithm colors the vertices
of an input graph using $O\left(a^{2}\cdot\log^{*}n\right)$ colors,
with a vertex-averaged complexity of $O\left(\log^{*}n\right)$ rounds.

\end{corollary}

For graphs with constant arboricity, we obtain the following corollary.

\begin{corollary}\label{O(log*n)VertexColoringAlgorithmCorollary}

For $k=\rho\left(n\right)$ and input graphs with arboricity $a=O\left(1\right)$,
the devised algorithm colors the vertices of an input graph using
$O\left(\log^{*}n\right)$ colors, with a vertex-averaged complexity
of $O\left(\log^{*}n\right)$ rounds.

\end{corollary}

\subsection{$O\left(ka\right)$-Vertex-Coloring in $O\left(a\log^{\left(k\right)}n\right)$
Vertex-Averaged Complexity\label{subsec:O(ka)VertexColoringSection}}

In this section we devise an algorithm that colors the vertices of
an input graph using $O\left(ka\right)$ colors within a vertex-averaged
complexity of $O\left(a\log^{\left(k\right)}n\right)$ rounds, for
$2\le k\le\rho\left(n\right)$. Moreover, for $k=\rho\left(n\right)$,
we obtain an $O\left(a\log^{*}n\right)$-vertex-coloring with a vertex
averaged complexity of $O\left(a\log^{*}n\right)$ rounds.

We obtain the algorithm by using the scheme of Section \ref{subsec:O(k * alpha)VertexColoringInImprovedVertexAveragedComplexity}.
As algorithm $\mathcal{A}$ of Section \ref{subsec:O(k * alpha)VertexColoringInImprovedVertexAveragedComplexity}
we use the $\left(\Delta+1\right)$-vertex-coloring algorithm of \cite{Barenboim2009}.
As algorithm $\mathcal{B}$ of Section \ref{subsec:O(k * alpha)VertexColoringInImprovedVertexAveragedComplexity}
we use the following algorithm.

Each vertex $v$ orients the edges incident on it using the following
logic. For each edge that connects the vertex $v$ with a neighbor
$u$ in the same $H$-set, which has a different color, the edge is
oriented towards the vertex with the higher color value.

As algorithm $\mathcal{C}$ of the scheme we use the following algorithm.
We recolor vertices as follows. Suppose we are currently looking at
vertices which belong to the sub-graph $G_{i}$ of $G$ (where $G_{i}$
is defined as in the proof of Lemma \ref{O(k * alpha)VertexColoringCorrectnessLemma}).
Each vertex first waits for all its parents with respect to the orientation
created in previous steps to assign themselves a new color from the
palette $\left\{ \left.\left(i-1\right)\left(A+1\right)+l\right|0\le l\le A-1\right\} $,
and then chooses a new color for itself from the same palette, different
than the color chosen by any of its parents.

We summarize the properties of the devised algorithm in the following
theorem.

\begin{theorem}\label{O(ak)DeterministicVertexColoringTheorem}

The devised algorithm properly colors the vertices of an input graph
using $O\left(ka\right)$ colors, with a vertex-averaged complexity
of $O\left(a\log^{\left(k\right)}n\right)$ rounds.

\end{theorem}

\begin{proof}

We first prove the devised algorithm computes a proper vertex-coloring
of the input graph using $O\left(ka\right)$ colors. We first observe,
that according to Section \ref{subsec:Procedure-Partition-Description},
the maximum degree of each sub-graph $G\left(H_{j}\right)$, induced
by an $H$-set $H_{j}$, is $O\left(a\right)$. Therefore, each invocation
of algorithm $\mathcal{A}$ on an $H$-set $H_{j}$ colors it using
$O\left(a\right)$ colors.

Next, we prove that algorithm $\mathcal{B}$, in this case, satisfies
the properties required by the scheme of Section \ref{subsec:O(k * alpha)VertexColoringInImprovedVertexAveragedComplexity},
for $\lambda=O\left(a\right)$. That is, we prove, that for each $H$-set
$H_{j}$ on which algorithm $\mathcal{B}$ is invoked, algorithm $\mathcal{B}$
produces an acyclic orientation of the edges of $G\left(H_{j}\right)$
with length $O\left(a\right)$ and out-degree $O\left(a\right)$.
The proof we present is based on Property 3.4 in \cite{Barenboim2013}.
We remind that in the invocation of algorithm $\mathcal{B}$ on the
edges of $G\left(H_{j}\right)$ we oriented each edge $e=\left\{ u,v\right\} $
towards the vertex with the larger color value. Therefore, following
this operation, each directed path in $G\left(H_{j}\right)$ consists
of vertices whose colors appear in strictly ascending order. Therefore,
the produced orientation of the edges of $G\left(H_{j}\right)$ is
acyclic. Also, each directed path contains at most $O\left(a\right)$
vertices, the same as the number of colors used by algorithm $\mathcal{A}$
to color $H_{j}$. Therefore, the length of the orientation is $O\left(a\right)$,
as required. As for the out-degree of the orientation, each vertex
in $H_{j}$ has $O\left(a\right)$ neighbors in $H_{j}$, according
to Section \ref{subsec:Procedure-Partition-Description}, and therefore,
in particular, has out-degree $O\left(a\right)$, as required.

We now prove that algorithm $\mathcal{C}$ computes a proper coloring
of the vertices of the input graph $G$. We observe that according
to Section \ref{subsec:Procedure-Partition-Description}, each vertex
$v\in H_{i}$, for some $1\le i\le\ell$, has at most $A$ neighbors
in $\cup_{j=i}^{\ell}H_{j}$. Since the color palette of algorithm
$\mathcal{C}$ consists of $\alpha=A+1$ colors, algorithm $\mathcal{C}$
properly colors each sub-graph $G_{i}$. 

We now analyze the vertex-averaged complexity of the devised algorithm.
The worst-case time complexity of algorithm $\mathcal{A}$ is known
to be $O\left(\Delta+\log^{*}n\right)$ \cite{Barenboim2009}. Since
the maximum degree of each $H$-set $H_{j}$ is $O\left(a\right)$,
the worst-time complexity of each invocation of algorithm $\mathcal{A}$
is in fact $O\left(a+\log^{*}n\right)$. The worst-case time complexity
of algorithm $\mathcal{B}$ is constant, since it involves only each
vertex locally orienting edges incident on it. 

The worst-case time complexity $T_{\mathcal{C},i}$ of invoking algorithm
$\mathcal{C}$ on a sub-graph $G_{i}$ is equal to the length of the
longest orientation in the sub-graph $G_{i}$. According to the above
analysis of algorithm $\mathcal{B}$, the length of the longest orientation
in each sub-graph $G\left(H_{j}\right)$ of $G$, for an $H$-set
$H_{j}$, following an invocation of algorithm $\mathcal{B}$, is
$O\left(a\right)$. Also, according to the proof of Lemma \ref{O(k * alpha)VertexColoringCorrectnessLemma},
the orientation on $G$ created by the scheme of Section \ref{subsec:O(k * alpha)VertexColoringInImprovedVertexAveragedComplexity}
is acyclic. Since the vertices of each sub-graph $G_{i}$ are partitioned
into $O\left(\log^{\left(i\right)}n\right)$ $H$-sets, it follows
that the length of the longest orientation in $G_{i}$ is $O\left(a\log^{\left(i\right)}n\right)$.
Therefore, it holds that $T_{\mathcal{C},i}=O\left(a\log^{\left(i\right)}n\right)$.
Therefore, according to Lemma \ref{O(k * alpha)VertexColoringVertexAveragedTimeLemma},
the vertex-averaged complexity of the devised algorithm is:
\begin{align*}
O\left(\log^{\left(k\right)}n+a+\log^{*}n+\sum_{i=1}^{k-1}\frac{a\log^{\left(i\right)}n}{\log^{\left(i\right)}n}+a\log^{\left(k\right)}n\right) & =\\
 & =O\left(a\log^{\left(k\right)}n\right)
\end{align*}
as required.

Therefore, according to Theorem \ref{O(k * alpha)VertexColoringSchemeTheorem},
the devised algorithm properly colors the input graph using $O\left(ka\right)$
colors, with a vertex-averaged complexity of $O\left(a\log^{\left(k\right)}n\right)$
rounds, as required.

\end{proof}

Suppose we set $k=\rho\left(n\right)$. By definition of $\rho\left(n\right)$,
it holds that $\log^{\left(k\right)}n=\log^{\rho\left(n\right)+1}n\le\log^{*}n$.
Therefore, for $k=\rho\left(n\right)$, we obtain the following corollary.

\begin{corollary}\label{O(alog*n)VertexColoringAlgorithmCorollary}

For $k=\rho\left(n\right)$, the devised algorithm colors the vertices
of an input graph using $O\left(a\log^{*}n\right)$ colors, with a
vertex-averaged complexity of $O\left(a\log^{*}n\right)$ rounds.

\end{corollary}

\subsection{$O(a^{1+\eta})$-Vertex-Coloring in $O(\log a\log\log n)$ Vertex-Averaged
Complexity}

\subsubsection{Background}

The algorithm presented in this section is based on Chapter 7 in \cite{Barenboim2013}
which contains a summary of results presented in \cite{Barenboim2011}.
Chapter 7 in \cite{Barenboim2013} presents methods for coloring an
input graph with a known arboricity $a$ in an improved worst-case
time complexity using at most a slightly asymptotically larger number
of colors compared to an earlier result from \cite{Barenboim2008}.
This last result obtained an $O(a)$-coloring of an input graph in
$O(a\log n)$ rounds. On the other hand, the result of Chapter 7 in
\cite{Barenboim2013} obtained an $O(a^{1+\eta})$-coloring in a worst-case
time complexity of $O(\log a\log n)$ rounds, for an arbitrarily small
constant $\eta>0$.

The above-mentioned result of \cite{Barenboim2008}, was obtained
by an algorithm called Procedure Arb-Color, whose properties are summarized
in Theorem 5.15 in \cite{Barenboim2008}. Procedure Arb-Color consists
of several steps. The first step consisted of executing Procedure
Partition and then orienting and labeling the edges of the graph in
a similar fashion to the method used in Section \ref{subsec:Parallelized-Forest-Decomposition}.
The second step consisted of coloring each subgraph $G(H_{i})$ induced
by a specific $H$-set $H_{i}$ for $i=1,2,...,\ell$ in parallel.
The final step consisted of performing a recoloring of the input graph's
vertices using a ``backward'' orientation starting from the vertices
of $H_{\ell}$ and ending in the vertices of $H_{1}$. The method
used for recoloring was as follows. Each parent first waits for each
of its children to choose a color from a certain palette. Afterwards,
the parent chooses an available color from the same palette.

The improved results of Chapter 7 in \cite{Barenboim2013} are based
upon the following procedures. The first procedure is Procedure Partial-Orientation.
For an input graph $G=(V,E)$, the procedure first executes Procedure
Partition on the vertices of $V$. The procedure then computes a defective
coloring of $G(H_{i})$ for each computed $H$-set $H_{i}$ (a coloring
where each vertex $v\in H_{i}$ may have up to a certain bounded number
of neighbors $d$ in $G(H_{i})$ with the same color as $v$). Subsequently,
the procedure uses the defective coloring to orient only edges connecting
vertices with differing colors, or which belong to different $H$-sets.
In addition to the input graph, this last procedure takes an integer
$t$ as an argument used to determine the parameters of the defective
coloring produced and used by the procedure. Ultimately, the procedure
produces an acyclic orientation of all edges of the input graph that
connect vertices in different $H$-sets, or that belong to the same
$H$-set but connect vertices with different colors. However, edges
connecting vertices that belong to the same $H$-set and have the
same color remain unoriented. For appropriate values of $t$, the
last procedure can produce orientations much shorter than the $O(a\log n)$
orientation of \cite{Barenboim2008}.

The second procedure is Procedure Arbdefective-Coloring. The procedure
takes three parameters as input. The first is the input graph, and
the latter two are two integers $k,t$. Procedure Arbdefective-Coloring
invokes Procedure Partial-Orientation and uses the produced partial
orientation to recolor the vertices, in order to produce an\emph{
$\left\lfloor \frac{a}{t}+\left(2+\epsilon\right)\frac{a}{k}\right\rfloor $}-\emph{arbdefective}
$k$-\emph{coloring} of $V$. 

A $b$-arbdefective $c$-coloring of a graph's vertices is an assignment
of one of $c$ colors to each vertex, such that the subgraph induced
by the vertices with a certain color has arboricity at most $b$.

Procedure Arbdefective-Coloring does this as follows. Each vertex
first waits for all of its parents in the orientation to choose a
color and then chooses a color used by the minimum number of parents.

The final procedure is Procedure Legal-Coloring which iteratively
partitions the input graph $G=(V,E)$ into smaller subgraphs $G_{i}$,
each induced by the set of vertices with a certain color $i$ and
each with a gradually smaller bounded arboricity. The number of times
that the graph is partitioned is decided based on a parameter $p$
passed to the procedure as an argument upon its invocation. Once the
arboricity of each subgraph created by the procedure becomes small
enough, all such subgraphs are colored in parallel using the recoloring
scheme used in the final step of the above-mentioned algorithm from
\cite{Barenboim2008}. A different unique color palette is used for
each subgraph based on a unique index given to each subgraph during
the execution of the procedure, ensuring a proper coloring. For completeness,
we present pseudo-code for Procedure Partial-Orientation, Procedure
Arbdefective-Coloring and Procedure Legal-Coloring as Algorithms \ref{alg:Procedure-Partial-Orientation}-\ref{alg:Procedure-Legal-Coloring}.

Finally, according to Corollary 4.6 in \cite{Barenboim2011}, for
an arbitrarily small constant $\eta>0$, invoking Procedure Legal-Coloring
with $p=2^{O\left(\frac{1}{\eta}\right)}$ produces an $O(a^{1+\eta})$-coloring
in $O(\log a\log n)$ deterministic time. We denote the last mentioned
algorithm, which produces an $a^{1+c}$-vertex-coloring for an arbitrarily
small constant $c>0$, as Procedure One-Plus-Eta-Legal-Coloring$(G,c)$.

\begin{algorithm}[H]
\begin{algorithmic}[1]
\STATE{$H_1, H_2,..., H_\ell\coloneqq$ an $H$-partition of $G$.}
\FOR{$i = 1,2,...,\ell$ in parallel}
\STATE{compute an $ \left\lfloor \frac{a}{t} \right\rfloor $-defective $O\left(t^2 \right)$-coloring $\psi_{i}$ of $G\left(H_{i}\right)$}
\ENDFOR
\FOR{each edge $e = (u,v)$ in $E$ in parallel} 
\IF{$u$ and $v$ belong to different $H$-sets} 
\STATE{orient $e$ towards the endpoint with greater $H$-index.} 
\ELSE 
\STATE{/* $u,v \in H_{i}$ for some $ i, 1 \le i \le \ell $ */} 
\IF{$u$ and $v$ have different colors} 
\STATE{orient $e$ towards the endpoint with greater $\psi_{i}$-color.} 
\ENDIF 
\ENDIF 
\ENDFOR
\end{algorithmic}

\caption{Procedure Partial-Orientation$\left(G,t\right)$ \cite{Barenboim2013}\label{alg:Procedure-Partial-Orientation}}
\end{algorithm}
\begin{algorithm}[H]
\begin{algorithmic}[1]
\STATE{$\mu = $ Partial-Orientation$\left( G,t \right)$}
\STATE{\textbf{once} all the parents $u$ of $v$ with respect to $\mu$ have selected a color $\phi\left(u\right)$, $v$ selects a color $\phi(v)$ from the palette \{1,2,...,$k$\} which is used by the minimum number of parents of $v$.}
\end{algorithmic}

\caption{Procedure Arbdefective-Coloring$\left(G,k,t\right)$ \cite{Barenboim2013}\label{alg:Procedure-Arbdefective-Coloring}}
\end{algorithm}
\begin{algorithm}[H]
\begin{algorithmic}[1]
\STATE{$G_{1} \coloneqq G$}
\STATE{$\alpha \coloneqq \alpha\left(G_{1}\right)$ /* $\alpha\left(G\right)$ is assumed to be known to all vertices. */}
\STATE{$\mathcal{G} \coloneqq {G_{1}}$ /* The set of subgraphs */}
\WHILE{$\alpha > p$}
\STATE{$\mathcal{\hat{G}} \coloneqq \emptyset$ /* Temporary variable for storing refinements of the set $\mathcal{G}$ */}
\FOR{each $G_{i} \in \mathcal{G}$ in parallel}
\STATE{$G_{1}^{'},G_{2}^{'},...,G_{p}^{'}$ $\coloneqq$ Arbdefective-Coloring$\left( G_{i}, k \coloneqq p, t \coloneqq p \right)$ 
\\
/* $G_{j}^{'}$ is the subgraph of $G_{i}$ induced by all the vertices that are assigned the color $j$ by the arbdefective coloring.*/}
\FOR{$j \coloneqq 1,2,...,p$ in parallel}
\STATE{$z \coloneqq \left( i-1 \right) \cdot p + j$ /* Computing a unique index for each subgraph. */}
\STATE{$\hat{G}_{z} \coloneqq G_{j}^{'}$}
\STATE{$\mathcal{\hat{G}} \coloneqq \mathcal{\hat{G}} \cup \left\{ \hat{G}_{z} \right  \}$}
\ENDFOR
\ENDFOR
\STATE{$\mathcal{G} \coloneqq \hat{\mathcal{G}}$}
\STATE{$\alpha \coloneqq \left\lfloor \frac{\alpha}{p} + \left( 2+\epsilon \right) \cdot \frac{\alpha}{p} \right\rfloor$
\\
/* The new upper bound for the arboricity of each of the subgraphs. */}
\ENDWHILE
\STATE{$A \coloneqq \left\lfloor \left( 2 + \epsilon \right) \alpha \right\rfloor + 1$}
\FOR{each $G_{i} \in \mathcal{G}$ in parallel}
\STATE{color $G_{i}$ legally using the palette $\left\{ \left( i-1 \right) \cdot A + 1, \left( i-1 \right ) \cdot A + 2,...,i \cdot A \right\}$
\\
/* Using Theorem 5.15 from \cite{Barenboim2013} */}
\ENDFOR
\end{algorithmic}

\caption{Procedure Legal-Coloring$\left(G,p\right)$ \cite{Barenboim2013}\label{alg:Procedure-Legal-Coloring}}
\end{algorithm}

\subsubsection{Improved Results}

In order to obtain an $O\left(a^{1+\eta}\right)$-vertex-coloring
algorithm with an improved vertex-averaged complexity, we devise a
new algorithm based on Procedure Legal-Coloring from \cite{Barenboim2011}.

We define a slightly modified version of Procedure Arbdefective-Coloring
from \cite{Barenboim2011}. Our version receives as input two integers
$k,t$, as in the original procedure, and an additional argument,
which is a partition of the input graph's vertices into $h$ $H$-sets
$\mathcal{H}=\left\{ H_{1},H_{2},...,H_{h}\right\} $. We henceforth
refer to this modified version of Procedure Arbdefective-Coloring
as Procedure H-Arbdefective-Coloring. The procedure computes an $\left\lfloor \frac{a}{t}+(2+\epsilon)\cdot\frac{a}{k}\right\rfloor $-arbdefective
$O(k)$-coloring of $\mathcal{H}$. The running time depends on the
number of $H$-sets in $\mathcal{H}$, rather than on $\log n$, as
in the original procedure.

We now present the algorithm for the main result of this section.
First, we define some notation. We denote the devised algorithm as
Procedure One-Plus-Eta-Arb-Col$(G,i,a)$, where the arguments of Procedure
One-Plus-Eta-Arb-Col are as follows. The argument $G$ is a graph
the algorithm receives as input, $i$ is an index given to the graph,
to track the current depth of the recursion, and $a$ is the arboricity
of $G$. The algorithm begins by invoking One-Plus-Eta-Arb-Col$(G,1,a)$,
where $G$ is the original input graph, and $a$ is the arboricity
of $G$.

Also, every vertex $v\in V$ begins with an ``empty'' color string,
consisting of 0 characters, and to which additional characters can
be appended. We now describe the steps of Procedure One-Plus-Eta-Arb-Col.\\
\\
One-Plus-Eta-Arb-Col($G,i,a$)
\begin{enumerate}
\item Let $C$ be a sufficiently large constant.
\item If $a<C$ then compute an $O(1)=O(a^{2})$-coloring of $G$ using
our $O(a^{2})$-coloring algorithm from Section \ref{subsec:O(ka^2)Coloring},
Theorem \ref{O(ka^2)VertexColoringTheorem}, using $k=2$.\label{enu:recursionBaseCaseColorStep}
\item Else:\label{enu:recursionNonBaseCaseStep}
\begin{enumerate}
\item Compute a partition of $V$ into the $H$-sets\\
$\mathcal{H}=\left\{ H_{1},H_{2},...,H_{r}\right\} $, for $r=\left\lceil 2\log\log n\right\rceil $.
Denote $H=\cup_{j=1}^{r}H_{j}$.\label{enu:createLoglognHsetsStep}
\item Invoke One-Plus-Eta-Legal-Coloring$\left(G(V\setminus H),\frac{1}{\log C}\right)$
of \cite{Barenboim2011} and add the prefix '1' to each resulting
color.\label{enu:rightHandSubgraphColoringStep}
\item In parallel to 3(a), 3(b), do \ref{enu:leftHandSubgraphArbdefectiveColStep}
followed by \ref{enu:parallelArbColInvocationStep}:\label{enu:leftHandSubgraphColoringStep}
\begin{enumerate}
\item Compute an $\frac{a}{C}$-arbdefective $O(C)$-coloring of $H$ by
invoking Procedure H-Arbdefective-Coloring with parameters $G(H)$,
$k=t=(3+\epsilon)C$, for $\epsilon=2$, and $\mathcal{H}$.\\
Let $G_{1},G_{2},....,G_{q}$ denote the resulting subgraphs induced
by colors $1,2,...,q=O(C)$, respectively. \label{enu:leftHandSubgraphArbdefectiveColStep}
\item For $j=1,2,...,q$ in parallel: Invoke One-Plus-Eta-Arb-Col($G_{j},i+1,\lfloor a/C\rfloor$),
and add the prefix $'2j'$ to each resulting color.\label{enu:parallelArbColInvocationStep}
\end{enumerate}
\end{enumerate}
\end{enumerate}
We prove the correctness of the procedure in the following lemma.

\begin{lemma}\label{arbColCorrectnessLemma}

For an input graph $G=(V,E)$, Procedure One-Plus-Eta-Arb-Col produces
a proper coloring of the vertices of $V$.

\end{lemma}

\begin{proof}

Let us denote by $b$ the maximum value the argument $i$ of the invocation
of Procedure One-Plus-Eta-Arb-Col takes in any invocation of Procedure
One-Plus-Eta-Arb-Col. We prove Procedure One-Plus-Eta-Arb-Col produces
a proper coloring of $G$ by induction on the value $l=b-i$.

In the base case, it holds that $l=b-b=0$. This means that for the
graph $G$ given as input to the current invocation of $A$, it holds
that $a<C$. Therefore, step \ref{enu:recursionBaseCaseColorStep}
is executed, which colors $G$ using $O(a^{2})$ colors, using the
algorithm from Section \ref{subsec:O(ka^2)Coloring}, Theorem \ref{O(ka^2)VertexColoringTheorem}.
By the correctness of the algorithm of Theorem \ref{O(ka^2)VertexColoringTheorem},
the coloring produced by Procedure One-Plus-Eta-Arb-Col in this case
is proper, as required.

Now, suppose that Procedure One-Plus-Eta-Arb-Col properly colors the
input graph $G$ for $l=k$, where $0<k<b$. We shall now prove that
Procedure One-Plus-Eta-Arb-Col properly colors the input graph $G$
for $l=k+1$.

Since $l=k+1>1$, in the current invocation of Procedure One-Plus-Eta-Arb-Col,
it holds that $a\ge C$. Therefore, in this case, step \ref{enu:recursionNonBaseCaseStep}
is executed. In particular , step 3(b) is executed. Following the
execution of step 3(b), by the correctness of Procedure One-Plus-Eta-Legal-Coloring,
shown in \cite{Barenboim2011}, the subgraph $G(V\setminus H)$\textbf{
}of $G$\textbf{ }is properly colored. Also, at the end of step 3(b),
we add the prefix $'1'$\textbf{ }to each resulting color, and at
the end of step \ref{enu:parallelArbColInvocationStep}, each vertex
$v\in H$ is assigned some color with a prefix of $'2'$. Therefore,
necessarily, for each edge $\{u,v\}\in E$, such that $u\in H$ and
$v\in V\setminus H$, the vertices $u,v$ have been assigned different
colors at the end of the execution of step \ref{enu:recursionNonBaseCaseStep}.

We now prove that for each edge $\{u,v\}\in E$, such that $u,v\in H$,
the vertices $u,v$ have different colors. In step \ref{enu:leftHandSubgraphArbdefectiveColStep}
we compute an $\frac{a}{C}$-arbdefective $O(C)$-coloring of the
vertices of $H$. Then, in step \ref{enu:parallelArbColInvocationStep},
for each subgraph $G_{j}$, for $j=1,2,...,q=O(C)$, induced by the
arbdefective coloring computed in step \ref{enu:leftHandSubgraphArbdefectiveColStep},
we invoke One-Plus-Eta-Arb-Col($G_{j},i+1,\lfloor a/C\rfloor$). For
each invocation of the form One-Plus-Eta-Arb-Col($G_{j},i+1,\lfloor a/C\rfloor$)
it holds that $l=b-(i+1)=k$. Therefore, by the inductive hypothesis,
each invocation One-Plus-Eta-Arb-Col($G_{j},i+1,\lfloor a/C\rfloor$),
carried out in step \ref{enu:parallelArbColInvocationStep}, produces
a proper coloring of the vertices of the subgraph $G_{j}$, induced
by the vertices of $H$ that were assigned the color $j$ by Procedure
H-Arbdefective-Coloring in step \ref{enu:leftHandSubgraphArbdefectiveColStep}.
Therefore, for each edge $\{u,v\}$, such that $u,v$ are both vertices
of some subgraph $G_{j}$, the vertices $u,v$ have different colors
at the end of step \ref{enu:parallelArbColInvocationStep}.

In addition, in step \ref{enu:parallelArbColInvocationStep}, for
each $j=1,2,...,q=O(C)$, for each vertex $v\in G_{j}$, after the
invocation One-Plus-Eta-Arb-Col($G_{j},i+1,\lfloor a/C\rfloor$) assigned
$v$ some color $c$, we added the prefix $'2j'$ to $c$. Therefore,
for each edge $\{u,v\}\in E$, where $u\in G_{j_{1}},v\in G_{j_{2}}$,
for $1\le j_{1}<j_{2}\le h$, it holds that $u,v$ have different
colors.

Overall, for each edge $\{u,v\}\in E$, it holds that $u,v$ have
different colors at the end of the current invocation of Procedure
One-Plus-Eta-Arb-Col. Therefore, for $l=k+1$, an invocation of the
form One-Plus-Eta-Arb-Col$(G,b-k-1,\lfloor a/C\rfloor)$ produces
a proper coloring of $G$. Therefore, for any $1\le i\le b$, the
invocation One-Plus-Eta-Arb-Col$(G,i,$$\lfloor a/C\rfloor$) produces
a proper coloring of $G$. Therefore, Procedure One-Plus-Eta-Arb-Col
produces a proper coloring of $G$, as required.

\end{proof}

Next, we analyze the number of colors used by Procedure One-Plus-Eta-Arb-Col.
We first present an intuitive explanation for computing the number
of colors used by Procedure One-Plus-Eta-Arb-Col.

Let us examine a recursive invocation of Procedure One-Plus-Eta-Arb-Col,
for which $a\ge C$. We observe that $G(V\setminus H)$ is colored
using $a^{1+\eta}$ colors, for a sufficiently small constant $\eta>0$.
In addition, we divide $G(H)$ into $q=O(C)$ subgraphs $G_{j}$,
each with arboricity at most $\left\lfloor \frac{a}{C}\right\rfloor $.
The number of colors we use to color $G(H)$ is $q$, multiplied by
the number of colors used by the recursive invocation of Procedure
One-Plus-Eta-Arb-Col on every subgraph $G_{j}$.

Therefore, the number of colors used by Procedure One-Plus-Eta-Arb-Col
is:
\[
f(a)=a^{1+\eta}+q\cdot f\left(\left\lfloor \frac{a}{C}\right\rfloor \right)
\]
Also, if $a<C$ in the current invocation of Procedure One-Plus-Eta-Arb-Col,
we color the current subgraph using $O(C^{2})$ colors. Thus, for
a sufficiently large value of $C$, the solution of the recursive
formula is:
\[
O\left(\left(\frac{q}{C}\right)^{\log_{C}a}\cdot a^{1+\eta}\right)
\]
If $\frac{q}{C}$ is constant and $C$ is sufficiently larger than
$\frac{q}{C}$, then:
\[
O\left(\left(\frac{q}{C}\right)^{\log_{C}a}\cdot a^{1+\eta}\right)=O\left(a^{1+\eta'}\right)
\]
for an arbitrarily small constant $\eta'>0$. Indeed, in our algorithm
$\frac{q}{C}=\frac{(3+\epsilon)C}{C}=3+\epsilon$, that is, $\frac{q}{C}$
is constant. Also, we can choose $C$ to be arbitrarily larger than
$3+\epsilon$, thus obtaining the required result. The analysis of
the number of colors employed by Procedure One-Plus-Eta-Arb-Col is
presented in the following lemma.

\begin{lemma}\label{arbColNumOfColsLemma}

For an input graph $G=(V,E)$, Procedure One-Plus-Eta-Arb-Col colors
the vertices of $V$ using $O\left(a^{1+\eta}\right)$ colors, for
an arbitrarily small constant $\eta>0$.

\end{lemma}

\begin{proof}

We first consider the maximum number of colors required by all graphs
$G(H)$. Let us denote $C_{1}=3+\epsilon$. In each execution of step
\ref{enu:leftHandSubgraphArbdefectiveColStep}, we further divide
the currently processed subgraph $G$, into $q=O(C)$ subgraphs $G_{1},G_{2},...,G_{q}$,
where each subgraph $G_{j}$, for $1\le j\le q$, has an arboricity
of at most $\frac{a}{C}$. After at most $\log_{C}a$ executions of
step \ref{enu:leftHandSubgraphArbdefectiveColStep}, we will reach
the state where $a<C$. At this point, we have divided the original
input graph given as argument to the invocation One-Plus-Eta-Arb-Col$(G,1,a)$
into at most $\left(C_{1}C\right)^{\log_{C}a}$ subgraphs. Each of
these subgraphs is colored using at most $C_{2}C^{2}$ colors in step
\ref{enu:recursionBaseCaseColorStep}, for an appropriate constant
$C_{2}>0$. Therefore, the executions of steps \ref{enu:recursionBaseCaseColorStep},3(c),
without taking into account yet the adding of the prefix $'2j'$ to
each color produced by the recursive invocation of Procedure One-Plus-Eta-Arb-Col
in step \ref{enu:parallelArbColInvocationStep}, increase the number
of colors used to color the subgraphs $G(H)$ by a factor of at most
$C_{2}C^{2}$.

Also, in each execution of step \ref{enu:parallelArbColInvocationStep},
to each subgraph $G_{j}$, for $j=1,2,..,q$, produced in step \ref{enu:leftHandSubgraphArbdefectiveColStep},
we append the prefix $'2j'$ to the color assigned to each vertex
of $G_{j}$ by invoking One-Plus-Eta-Arb-Col($G_{j},i+1,\lfloor a/C\rfloor$).
Again, after at most $\log_{C}a$ invocations of step 3(c), we will
reach the state where $a<C$ and execute step \ref{enu:recursionBaseCaseColorStep}
once. Therefore, the maximum number of times a prefix of the form
$'2j'$ is added to a color is $\log_{C}a$ times, for $j=1,2,...,q=C_{1}C$.
In particular, appending the character $'2'$ as a prefix to some
color at most $\log_{C}a$ times increases the number of colors used
to color the subgraphs $G(H)$ by a factor of at most $2^{\log_{C}a}$.
In addition, adding the prefix $'j'$ to some color at most $\log_{C}a$
times increases the number of colors used to color the subgraphs $G(H)$
by a factor of at most $\left(C_{1}C\right)^{\log_{C}a}$.

Therefore, the adding of the prefix $'2j'$ to the color produced
by the recursive invocation of Procedure One-Plus-Eta-Arb-Col increases
the number of colors used to color the subgraphs $G(H)$ by a factor
of at most:
\[
2^{\log_{C}a}\cdot\left(C_{1}C\right)^{\log_{C}a}
\]
Overall, the number of colors used to color the subgraphs $G(H)$
is at most:
\begin{align*}
C_{2}C^{2}\cdot2^{\log_{C}a}\cdot\left(C_{1}C\right)^{\log_{C}a} & =C_{2}C^{2}\cdot a^{\frac{1}{\log C}}\cdot2^{\log_{C}a\log\left(C_{1}C\right)}\\
 & =C_{2}C^{2}\cdot a^{\frac{1}{\log C}}\cdot a^{\frac{\log\left(C_{1}C\right)}{\log C}}\\
 & =C_{2}C^{2}\cdot a^{\frac{1}{\log C}}\cdot a^{\frac{\log C_{1}}{\log C}+1}\\
 & =C_{2}C^{2}\cdot a^{1+\frac{\log C_{1}+1}{\log C}}
\end{align*}
We now look at the maximum number of colors used to color the subgraphs
$G(V\setminus H)$. These subgraphs are colored in the execution of
step 3(b). In this step, the subgraph $G\left(V\setminus H\right)$
is colored using the invocation One-Plus-Eta-Legal-Coloring$\left(G(V\setminus H),\frac{1}{\log C}\right)$,
using at most $C_{3}a^{1+\frac{1}{\log C}}$, for some sufficiently
large constant $C_{3}$. Then, to each resulting color the prefix
$'1'$ is added. Therefore, the maximum number of colors used to color
the subgraphs $G\left(V\setminus H\right)$ is at most:
\begin{align*}
2^{\log_{C}a}\cdot C_{3}a^{1+\frac{1}{\log C}} & =a^{\frac{1}{\log C}}\cdot C_{3}a^{1+\frac{1}{\log C}}\\
 & =C_{3}a^{1+\frac{2}{\log C}}
\end{align*}
Let us set $C_{4}=\max\left\{ C_{2}C^{2},C_{3}\right\} $. Overall,
the maximum number of colors used by Procedure One-Plus-Eta-Arb-Col
to color an input graph $G$ is at most:
\begin{align*}
2C_{4}a^{1+\frac{\log C_{1}+3}{\log C}} & =2C_{4}a^{1+\frac{\log\left(3+\epsilon\right)+3}{\log C}}\\
 & <2C_{4}a^{1+\frac{\log8+3}{\log C}}\\
 & =2C_{4}a^{1+\frac{6}{\log C}}
\end{align*}
Let us denote $\eta=\frac{6}{\log C}$. It follows that for a sufficiently
large constant $C>0$, the maximum number of colors used by Procedure
One-Plus-Eta-Arb-Col is $O\left(a^{1+\eta}\right)$, as required.

\end{proof}

\begin{lemma}\label{arbColVertexAveragedComplexityLemma}

For an input graph $G=(V,E)$, Procedure One-Plus-Eta-Arb-Col colors
the vertices of $V$ with a vertex-averaged complexity of $O(\log a\log\log n)$
rounds.

\end{lemma}

\begin{proof}

The worst-case time complexity required by coloring each subgraph
$G(V\setminus H)$, using Procedure One-Plus-Eta-Legal-Coloring, requires
at most the worst-case time complexity of executing Procedure One-Plus-Eta-Legal-Coloring
on the original graph $G$ given as input to the invocation One-Plus-Eta-Arb-Col$(G,1,a)$.
The worst-case time complexity of Procedure One-Plus-Eta-Legal-Coloring,
when executed on an input graph $G=(V,E)$ with arboricity $a$, is
$O(\log a\log n)=O(\log^{2}n)$ rounds. The last transition follows
from the fact that $a\le n$.

Also, we only execute Procedure One-Plus-Eta-Legal-Coloring on a subgraph
$G\left(V\setminus H\right)$ induced by the $H$-sets $H_{r+1},...,H_{O(\log n)}$,
which were produced by executing Procedure Partition on $G$. We remind
that we set $\epsilon=2$. Similarly to the analysis of the vertex-averaged
complexity of Procedure Partition, the number of vertices that take
part in this invocation of Procedure One-Plus-Eta-Legal-Coloring is
at most:
\begin{align*}
\sum_{i=\left\lceil 2\log\log n\right\rceil +1}^{\left\lceil \log_{\frac{2+\epsilon}{2}}n\right\rceil }n_{i} & \le\sum_{i=2\log\log n+2}^{\log_{\frac{2+\epsilon}{2}}n+1}\left(\frac{2}{2+\epsilon}\right)^{i-1}n\\
 & =\sum_{i=2\log\log n+2}^{\log n+1}\left(\frac{1}{2}\right)^{i-1}n\\
 & =\sum_{i=0}^{\log n-2\log\log n-1}\left(\frac{1}{2}\right)^{i+2\log\log n+1}n\\
 & \le n\cdot\left(\frac{1}{2}\right)^{2\log\log n+1}\cdot\frac{1-\left(\frac{1}{2}\right)^{\log n}}{1-\frac{1}{2}}\\
 & =\frac{n}{2^{\log\log^{2}n}}\cdot O(1)\\
 & =O\left(\frac{n}{\log^{2}n}\right)
\end{align*}
In addition, as explained above, the maximum depth of the recursion,
when invoking One-Plus-Eta-Arb-Col$(G,1,a)$, is at most:
\[
\log_{C}a=\frac{\log a}{\log C}=\frac{1}{\log C}\log a=O(\log a)
\]
For all recursive invocations One-Plus-Eta-Arb-Col$(G,i,a)$, with
the same value of $i$, the coloring of each subgraph $G(V\setminus H)$
is carried out in parallel. Therefore, the sum of the number of rounds
carried out by all vertices executing step 3(b) throughout the execution
of Procedure One-Plus-Eta-Arb-Col on the original input graph is at
most: 
\begin{align*}
O(\log^{2}n)\cdot O\left(\frac{n}{\log^{2}n}\right)\cdot O(\log a) & =O\left(n\log a\right)
\end{align*}
We now analyze the worst-case time complexity of coloring each subgraph
$G(H)$. The recursive invocation of Procedure One-Plus-Eta-Arb-Col
stops when $a<C$. In this case, step \ref{enu:recursionBaseCaseColorStep}
is executed. In step \ref{enu:recursionBaseCaseColorStep}, we compute
an $O(a^{2})$ coloring of each subgraph $G\left(H\right)$ in parallel
using the algorithm of Theorem \ref{O(ka^2)VertexColoringTheorem},
which has a vertex-averaged complexity of $O(\log\log n)$. Also,
the number of vertices executing the coloring of the subgraphs $G(H)$
can trivially be at most $n$. Therefore, the sum of the number of
rounds of communication executed by each vertex taking part in the
execution of step \ref{enu:recursionBaseCaseColorStep} is $O(n\log\log n)$.

We now analyze the worst-case time complexity of each execution of
Procedure H-Arbdefective-Coloring. In our algorithm, we invoked Procedure
H-Arbdefective-Coloring with the integer parameters $k=t=\left(3+\epsilon\right)C$.
Following a similar analysis to that of the worst-case time complexity
of Procedure Arbdefective-Coloring in \cite{Barenboim2011}, the worst-case
time complexity of Procedure H-Arbdefective-Coloring is $O(t^{2}\left|\mathcal{H}\right|)=O(r)$.
Since $\left|\mathcal{H}\right|=r=O(\log\log n)$, the worst-case
time complexity of a single execution of Procedure H-Arbdefective-Coloring
in our algorithm is $O(\log\log n)$.

As explained above, the maximum depth of the recursion, when invoking
One-Plus-Eta-Arb-Col$(G,1,a)$, is $O(\log a)$. This is also the
maximum number of times any vertex $v\in V$ executes Procedure H-Arbdefective-Coloring.
Since the number of vertices executing Procedure H-Arbdefective-Coloring
can trivially be at most $n$, the sum of the number of rounds of
communication carried out by all vertices executing Procedure H-Arbdefective-Coloring
throughout the execution of Procedure One-Plus-Eta-Arb-Col is $O(n\log a\log\log n)$.
Therefore, the vertex-averaged complexity of the execution of Procedure
One-Plus-Eta-Arb-Col on an input graph $G$ is:
\[
\frac{O\left(n\log a\right)+O\left(n\log a\log\log n\right)}{n}=O\left(\log a\log\log n\right)
\]
as required.

\end{proof}

The following theorem summarizes the properties of Procedure One-Plus-Eta-Arb-Col.

\begin{theorem}

Procedure One-Plus-Eta-Arb-Col computes a proper $O\left(a^{1+\eta}\right)$-vertex-coloring
of an input graph $G=(V,E)$ with a vertex-averaged complexity of
$O(\log a\log\log n)$ rounds, for an arbitrarily small constant $\eta>0$.

\end{theorem}

\section{Solving Problems of Extension from any Partial Solution in Improved
Vertex-Averaged Complexity\label{sec:Solving-Problems-of-Extension-From-Any-Partial-Solution}}

\subsection{General Method}

In this section we define a class of problems we name \emph{problems
of extension} \emph{from any partial solution}. Also, we devise a
general method for converting an algorithm for a problem from this
class with a worst-case time complexity of $f(\Delta,n)$ to an algorithm
with a vertex-averaged complexity of $f(a,n)$. We begin with some
definitions.

\begin{definition}

Suppose we are given a graph $G=(V,E)$ and a problem $P$. Then $P$
is a problem of extension from any partial solution, if for any subgraph
$H'=(V',E')$, $V'\subseteq V$, $E'\subseteq E$, with a proper solution
$S'$ to $P$, there exists an algorithm $\mathcal{A}$ that can compute
a solution $S$ for $P$ on $G$ without changing the solution $S'$
for $H'$.

We refer to such an algorithm $\mathcal{A}$ as an algorithm for a
problem of extension from any partial solution.

\end{definition}

We observe that the main symmetry-breaking problems the paper deals
with, vertex-coloring, MIS, edge-coloring and maximal matching, are
all problems of this class. We also note that in \cite{Feuilloley2017a},
the author refers to the set of languages corresponding to the class
of problems of extension from any partial solution as \emph{completable}
LCL{*}\emph{ languages}, where LCL{*} is defined by the same author
in \cite{Feuilloley2017a} as the set of languages $\mathcal{L}$,
for which there exists a constant-time distributed verification algorithm
that accepts at every node of a given graph if and only if the graph
belongs to the language $\mathcal{L}$. (This is an extension of the
term \emph{locally checkable labelings }(LCL), used to describe problems
where the validity of the solution can be verified within a constant
number of rounds. The term LCL, to the best of our knowledge, was
first presented in \cite{Naor1995}). An early paper related to LCL
and problems of extension from any partial solution is \cite{Kutten1995}.
In \cite{Kutten1995}, the authors dealt with methods for transforming
an incorrect solution for certain problems, such as MIS, into a correct
one, in a time complexity that depends on the number of faulty vertices
in the graph, rather than, for example, on the number of all vertices
in the graph, which can potentially be much larger.

We now present a general method to convert an algorithm which solves
a problem of extension from any partial solution with a worst-case
time complexity given as a function of $\Delta,n$ to another algorithm
for the same problem, with a vertex-averaged complexity given as a
function of $a,n$.

\begin{theorem}\label{algorithmForPoEMethodTheorem}

Suppose we are given a problem of extension from any partial solution
$P$ and an algorithm for a problem of extension from any partial
solution $\mathcal{A}$ for solving $P$ with a worst-case time complexity
of $f(\Delta,n)$. Then, there exists an algorithm $\mathcal{A}'$
for solving $P$ with a vertex-averaged complexity of $O\left(f(a,n)\right)$.

\end{theorem}

\begin{proof}

We prove this theorem by presenting a general method to convert algorithm
$\mathcal{A}$ to a different algorithm $\mathcal{A}'$ as described.
Suppose that our algorithms are invoked on some input graph $G=(V,E)$.
The conversion method accepts as input algorithm $\mathcal{A}$ and,
if necessary, an additional algorithm $\mathcal{B}$. Algorithm $\mathcal{B}$
is a general notation for an algorithm that is used in the case of
problems, where a label needs to be assigned to each edge, such as
edge-coloring and maximal matching. We say that it is \emph{necessary
}to invoke such an algorithm $\mathcal{B}$ if and only if the problem
we are trying to solve is concerned with assigning a label to each
edge. If an algorithm is used as algorithm $\mathcal{B}$, we require
algorithm $\mathcal{B}$ to satisfy several properties. The motivation
for the use of an algorithm $\mathcal{B}$ is better clarified in
the following outline of the produced algorithm $\mathcal{A}'$.

Algorithm $\mathcal{A}'$ performs $\ell=O(\log n)$ iterations. In
each iteration $1\le i\le\ell$, it invokes algorithm $\mathcal{A}$
on $G\left(H_{i}\right)$. If it is not necessary to invoke algorithm
$\mathcal{B}$, then by definition of $\mathcal{A}$, the solution
computed so far for $G\left(\cup_{j=1}^{i}H_{j}\right)$ is proper.
Otherwise, for $i\ge2$, at this point, edges crossing from $H_{i}$
to $\cup_{j=1}^{i-1}H_{j}$ will remain unhandled. Therefore, we invoke
the input algorithm $\mathcal{B}$. We require algorithm $\mathcal{B}$
to satisfy the following properties. First, algorithm $\mathcal{B}$
needs to satisfy that at the end of its execution, the solution computed
so far for $P$, on $G\left(\cup_{j=1}^{i}H_{j}\right)$ is proper.
Second, its worst-case time complexity needs to be $O\left(f(a,n)\right)$,
to help obtain the improved vertex-averaged complexity. Lastly, to
help obtain the improved vertex-averaged complexity, algorithm $\mathcal{B}$
must be executed only by the vertices of $H_{i}$. Examples of the
usage of an algorithm $\mathcal{B}$ can be found in Section \ref{subsec:PoEApplications}.
Next, for each $1\le i\le\ell$, let $G\left(H_{i}\right)=\left(H_{i},E_{i}\right)$
where $E_{i}=\left\{ e=\left\{ u,v\right\} \in E|u,v\in H_{i}\right\} $.
We define the steps of the conversion method more formally as follows:
\begin{enumerate}
\item Execute Procedure Parallelized-Forest-Decomposition.
\item Within the invocation of Procedure Parallelized-Forest-Decomposition
in step 1, we perform the following steps in each iteration $i$ of
the procedure's main loop:
\begin{enumerate}
\item We compute an $H$-set $H_{i}$ and decompose all edges of $E_{i}$
into oriented forests. 
\item We invoke algorithm $\mathcal{A}$ on the subgraph $G(H_{i})$ of
$G$.
\item In addition, for $i\ge2$, if necessary, we invoke another algorithm
$\mathcal{B}$ as described above on the sub-graph induced by all
edges $\{u,v\}\in E$, such that $u\in\cup_{j=1}^{i-1}H_{j}$, $v\in H_{i}$.
\end{enumerate}
\item We note that algorithm $\mathcal{A}$ is invoked on $H_{i+1}$ only
after all of the following events have occurred:
\begin{enumerate}
\item Algorithm $\mathcal{A}$ has completed execution for $H_{i}$.
\item Algorithm $\mathcal{B}$ has completed execution for $H_{i}$.
\item The $H$-set $H_{i+1}$ has been formed by iteration $i+1$ of Procedure
Parallelized-Forest-Decomposition\emph{.}
\item All edges of $E_{i+1}$ have been decomposed into $O(a)$ oriented
forests.
\end{enumerate}
\end{enumerate}
First, we prove the correctness of this general method. Namely, we
prove that algorithm $\mathcal{A}'$ computes a proper solution for
$P$, on $G$.

In each iteration $i$ of the main loop of algorithm $\mathcal{A}'$
we invoke algorithm $\mathcal{A}$ on $G\left(H_{i}\right)$. Also,
according to the definition of $\mathcal{A}$, we do not change the
partial solution computed so far on the subgraph $G\left(\cup_{j=1}^{i-1}H_{j}\right)$.
If it is not necessary to execute an algorithm $\mathcal{B}$, the
solution computed so far for $G\left(\cup_{j=1}^{i}H_{j}\right)$
is proper, by the definition of algorithm $\mathcal{A}$. Otherwise,
for $i>1$, we invoke algorithm $\mathcal{B}$ on the subgraph induced
by the edges $\{u,v\}$, such that $u\in\cup_{j=1}^{i-1},v\in H_{i}$.
By the definition of algorithm $\mathcal{B}$, at the end of the execution
of algorithm $\mathcal{A}'$, the solution produced so far for $P$,
on $G\left(\cup_{j=1}^{i}H_{j}\right)$, is proper. It easily follows
that the solution produced at the end of the execution of algorithm
$\mathcal{A}'$ for $P$ on $G$ is proper, as required.

We now prove that algorithm $\mathcal{A}'$ has a vertex-averaged
complexity of $O\left(f(a,n)\right)$. As explained in Section \ref{subsec:Procedure-Partition-Description},
for a subgraph $G\left(H_{i}\right)$ of $G$ for some $H$-set $H_{i}$,
for any vertex $v\in H_{i}$ it holds that the degree of $v$ is $O(a)$.

Therefore, the execution of algorithm $\mathcal{A}$ in each iteration
of the main loop of algorithm $\mathcal{A}'$ has a worst-case running
time of $O\left(f(a,n)\right)$. Also, if we invoke some procedure
as algorithm $\mathcal{B}$, its worst-case time complexity is by
definition $O\left(f\left(a,n\right)\right)$. Therefore, it follows
from Corollary \ref{partitionBasedAlgAverageTimePerVertex} that the
vertex-averaged complexity of algorithm $\mathcal{A}'$ is $O\left(f(a,n)\right)$,
as required.

\end{proof}

\subsection{Applications to Specific Problems of Extension from any Partial Solution\label{subsec:PoEApplications}}

In this section we present several corollaries. Each of the corollaries
constitutes an application of Theorem \ref{algorithmForPoEMethodTheorem}.
Within the section, we use $\mathcal{A}$,$\mathcal{B}$ to denote
the algorithms that are also denoted as $\mathcal{A}$,$\mathcal{B}$
in Theorem \ref{algorithmForPoEMethodTheorem} and $\mathcal{A}'$
to denote the algorithm generated by the method of Theorem \ref{algorithmForPoEMethodTheorem}.

The first corollary we present obtains an improved vertex-averaged
complexity for $(\Delta+1)$-vertex-coloring. We obtain this corollary
as follows, for a given graph $G=(V,E)$. We denote by $\deg_{G}(v)$
the degree of a vertex $v\in V$ in the graph $G$. Also, we define
the problem \emph{$(\deg+1)$-list-coloring} as follows. Each vertex
$v\in V$ is assigned a list of colors $L(v)\subseteq\mathcal{C}$,
for some finite set $\mathcal{C}$, such that $\left|L(v)\right|\ge\deg_{G}(v)+1$.
Then, we need to assign a color $c_{v}\in L(v)$ to each vertex $v\in V$,
that is different than the color of each neighbor of $v$.

We also employ a result of \cite{Fraigniaud2016}. Specifically, Theorem
4.1 of \cite{Fraigniaud2016} implies that $(\deg+1)$-list-coloring
can be solved in $O(\sqrt{\Delta}\log^{2.5}\Delta)$ rounds in the
worst-case. Henceforth, we refer to this implied algorithm as Procedure
Deg-Plus1-List-Col.

We now devise our algorithm for computing a $(\Delta+1)$-vertex-coloring
in improved vertex-averaged complexity. We denote this algorithm as
algorithm $\mathcal{A}'$. We obtain algorithm $\mathcal{A}'$ by
using the following procedure as algorithm $\mathcal{A}$. Initially,
each vertex $v\in V$ is given a list of colors, as in the case of
$(\deg+1$)-list-coloring, with $\mathcal{C}=\{1,2,...,\Delta+1\}$.
In each iteration $i$ of algorithm $\mathcal{A}'$, we invoke Procedure
Deg-Plus1-List-Col on the subgraph $G(H_{i})$. Subsequently, each
vertex $v$ sends the color it assigned itself to each neighbor $u\in\left(V\setminus\cup_{j=1}^{i}H_{j}\right)$
and then $u$ removes the color of $v$ from $L(u)$. The next corollary
summarizes the properties of the devised algorithm.

\begin{corollary}\label{vertexColPoECorollary}

For a given graph $G=(V,E)$, the vertices of $V$ can be colored
using $\Delta+1$ colors with a vertex-averaged complexity of $O(\sqrt{a}\log^{2.5}a+\log^{*}n)$
rounds.

\end{corollary}

\begin{proof}

First, we prove that the algorithm $\mathcal{A}$ properly colors
the vertices of $H_{i}$. For each edge $\{u,v\}\in E$, where $u,v\in H_{i}$,
the vertices $u,v$ necessarily have different colors, by the correctness
of Procedure Deg-Plus1-List-Col, according to \cite{Fraigniaud2016}.

Suppose now that $i\ge2$. Let us look at an edge $e=\{u,v\}\in E$,
where $u\in H_{j},v\in H_{i}$, for $j<i$. At the end of iteration
$j$ of algorithm $\mathcal{A}'$, $u$ sent its color to $v$, and
subsequently, $v$ removed the color from $L(v)$. At the beginning
of each iteration $i$ of algorithm $\mathcal{A}'$, for each vertex
$v\in V$, at most $\deg_{G}(v)-\deg_{G(H_{i})}(v)$ colors have been
removed from $L(v)$. Therefore, in iteration $i$ of algorithm $\mathcal{A}'$,
it holds that $\left|L(v)\right|\ge\deg_{G(H_{i})}(v)+1$. Therefore,
an available color exists for coloring $v$ and it is different than
that of $u$. Therefore, algorithm $\mathcal{A}'$ produces a proper
coloring, as required.

As for the number of colors used by algorithm $\mathcal{A}'$, all
colors are taken from the set $\{1,2,...,\Delta+1\}$, so the total
number of colors used is $\Delta+1$, as required.

We now analyze the worst-case time complexity of algorithm $\mathcal{A}$.
The worst-case time complexity of algorithm $\mathcal{A}$ is asymptotically
at most the worst-case time complexity of the execution of Procedure
Deg-Plus1-List-Col, when executed on a subgraph induced by an $H$-set
$H_{i}$. Therefore, based on the discussion preceding this corollary,
the worst-case time complexity of algorithm $\mathcal{A}$ is $O(\sqrt{a}\log^{2.5}a+\log^{*}n)$.

Therefore, by Theorem \ref{algorithmForPoEMethodTheorem}, the vertex-averaged
complexity of algorithm $\mathcal{A}'$ is $O(\sqrt{a}\log^{2.5}a+\log^{*}n)$,
as required.

\end{proof}

We now present an algorithm that obtains improved vertex-averaged
complexity for MIS. We denote by $MIS(H)$ an MIS for a certain graph
$H$. In addition, for $1\le i\le\ell$ we denote by $MIS_{i}$ a
set of vertices $v\in H_{i}$ satisfying the following properties.
The first property is that for each vertex $u\in\left(H_{i}\setminus MIS_{i}\right)$,
$u$ has some neighbor $v\in\cup_{j=1}^{i}MIS_{j}$. The second property
is that for each vertex $v\in MIS_{i}$, $v$ has no neighbors in
$\cup_{j=1}^{i}MIS_{j}$. From these properties, it follows that $\cup_{j=1}^{i}MIS_{j}$
is an MIS for $G\left(\cup_{j=1}^{i}H{}_{j}\right)$. In particular,
for $i=1$, $MIS_{1}$ is simply an MIS for $G(H_{1})$.

We achieve the result stated by the corollary by using the following
algorithm as algorithm $\mathcal{A}$, in each iteration $i$ of algorithm
$\mathcal{A}'$. The algorithm consists of two steps. Initially, $MIS(G)=\emptyset$. 

In the first step, we compute a proper $\Delta\left(G(H_{i})\right)+1$
coloring of the vertices of $H_{i}$ using Procedure Deg-Plus1-List-Col
where each vertex $v\in V$ is given an initial list of colors $L(v)=\{1,2,...,\deg_{G(H_{i})}(v)+1\}$.
In the second step, we compute $MIS_{i}$ as follows. We execute a
loop for $j=1,2,...,\Delta\left(G(H_{i})\right)+1$. In each iteration
$j$ of this loop, each vertex $v\in H_{i}$ with color $j$ that
has no neighbor in $\cup_{l=1}^{i}MIS_{l}$ adds itself to $MIS_{i}$.
The devised algorithm corresponds to the following corollary.

\begin{corollary}\label{MISPoECorollary}

For a graph $G=(V,E)$, an MIS can be computed for $G$ with a vertex-averaged
complexity of $O(a+\log^{*}n)$ rounds.

\end{corollary}

\begin{proof}

We first prove the correctness of algorithm $\mathcal{A}'$. This
is done by first proving that in each iteration $i$ of $\mathcal{A}'$,
after computing $MIS_{i}$, it holds that $\cup_{j=1}^{i}MIS_{j}$
is an MIS for $G\left(\cup_{j=1}^{i}H_{j}\right)$. We prove this
by induction on the value of the current iteration $i$ of $\mathcal{A}'$.

For the base case $i=1$, in the second step of the invocation of
algorithm $\mathcal{A}$, we are in fact executing the reduction from
MIS to $(\Delta+1)$-vertex-coloring, described in Section 3.2 of
\cite{Barenboim2013}, for the subgraph $G(H_{i})$ of $G$ induced
by $H_{i}$. Therefore, by the correctness of this reduction, $MIS_{i}$
is in this case, an MIS for $G(\cup_{j=1}^{1}H_{j})=G(H_{1})$, as
required.

Suppose now, that for $i\le k<\ell$, $\cup_{j=1}^{i}MIS_{j}$ is
an MIS for $G\left(\cup_{j=1}^{i}H_{j}\right)$. We now prove that
$\cup_{j=1}^{k+1}MIS_{j}$ is an MIS for $G\left(\cup_{j=1}^{k+1}H_{j}\right)$.
To this end, we prove that at the end of iteration $i=k+1$ of algorithm
$\mathcal{A}'$, the set $MIS_{k+1}$ satisfies the two properties
mentioned in the beginning of the proof, that a set $MIS_{i}$ needs
to satisfy.

As for the first property, let us suppose for contradiction that there
exists a vertex $u\in\left(H_{i}\setminus MIS_{i}\right)$, such that
$u$ has no neighbor in $\cup_{l=1}^{i}MIS_{l}$. However, this means
that in the appropriate iteration $j$ of the loop in the second step
of algorithm $\mathcal{A}$, $u$ would have checked and found that
it had no neighbor in $\cup_{l=1}^{i}MIS_{l}$. Subsequently, the
vertex $u$ would have added itself to $MIS_{i}$, contradicting that
$u\in\left(H_{i}\setminus MIS_{i}\right)$.

As for the second property, let us assume for contradiction that there
exists a vertex $v\in MIS_{i}$ that has some neighbor $u\in MIS_{l}$
for some $l\le i$. Also, suppose that during the execution of the
first step of algorithm $\mathcal{A}$, $u$ was colored with some
color $j_{1}$, while $v$ was colored with some color $j_{2}$. Without
loss of generality, suppose that $j_{2}>j_{1}$. Therefore, during
the execution of the loop of the second step of algorithm $\mathcal{A}$,
the vertex $v$ would have found out in iteration $j_{2}$ that $u\in MIS_{l}$
and would have not joined $MIS_{i}$, leading to a contradiction.

Therefore, the two above-mentioned properties, that a set $MIS_{i}$
needs to satisfy, hold for $i=k+1$. Therefore, they hold for any
$1\le i\le\ell$. In particular, they hold for $i=\ell$. That is,
$\cup_{j=1}^{\ell}MIS_{j}$ is a set of vertices, such that every
vertex in $\cup_{j=1}^{\ell}H_{j}\setminus\cup_{j=1}^{\ell}MIS_{j}$
has some neighbor in $\cup_{j=1}^{\ell}MIS_{j}$, and $\cup_{j=1}^{\ell}MIS_{j}$
is independent. Therefore, $\cup_{j=1}^{\ell}MIS_{j}$ is an MIS for
$G$, as required.

We now analyze the worst-case time complexity of algorithm $\mathcal{A}$.
Using a similar analysis to that carried out in the proof of Corollary
\ref{vertexColPoECorollary}, the worst-case time complexity of the
first step of algorithm $\mathcal{A}$ is $O(\sqrt{a}\log^{2.5}a+\log^{*}n)$.

We now analyze the worst-case time complexity of the second step of
algorithm $\mathcal{A}$. In the second step, we execute a loop consisting
of a number of iterations equal to the number of colors used to color
the vertices of $H_{i}$ in the first step of algorithm $\mathcal{A}$.
Also, each iteration of the loop requires constant time. It can be
easily verified that the number of colors used is at most $\Delta(G(H_{i}))+1$.
According to Section \ref{subsec:Procedure-Partition-Description},
it holds that $\Delta(G(H_{i}))=O(a)$. Therefore, the worst-case
time complexity of algorithm $\mathcal{A}$ is $O(\sqrt{a}\log^{2.5}a+a+\log^{*}n)$.\\
Therefore, by Theorem \ref{algorithmForPoEMethodTheorem}, the vertex-averaged
complexity of algorithm $\mathcal{A}'$ is $O(a+\log^{*}n)$.

\end{proof}

Another corollary follows from Corollary \ref{MISPoECorollary} as
follows.

\begin{corollary}\label{O(1)ArbMISPoECorollary}

For a graph $G=(V,E)$ with a constant arboricity $a$, an MIS can
be computed with a vertex-averaged complexity of $O(\log^{*}n)$ rounds.

\end{corollary}

We now present an algorithm which obtains improved vertex-averaged
complexity for $\left(2\Delta-1\right)$-edge-coloring. We obtain
the result by using the following algorithms as algorithms $\mathcal{A}$,$\mathcal{B}$
in each iteration $i$ of algorithm $\mathcal{A}'$. The algorithm
we use as algorithm $\mathcal{A}$ is as follows. We invoke a deterministic
$\left(2\Delta-1\right)$-edge-coloring algorithm from \cite{Panconesi2001}
on $G(H_{i})$. 

We now describe the algorithm used as algorithm $\mathcal{B}$. First,
we note that in the beginning of each iteration $i$, algorithm $\mathcal{A}'$
oriented all edges $e=\{u,v\}\in E$, such that $u\in H_{i},v\in V\setminus\cup_{k=1}^{i-1}H_{k}$.
Also, each vertex $v\in H_{i}$ labeled each of the outgoing edges
incident on it with a different label from $\{1,2,...,A\}$. Next,
we describe the operations we perform in algorithm $\mathcal{B}$. 

We execute a loop for $j=1,2,...,A$. We denote by $G_{j}(v)$ the
star subgraph of $G$ induced by all edges $e=\{u,v\}$, such that
$u\in\cup_{k=1}^{i-1}H_{j},v\in H_{i}$ and $e$ is labeled $j$.
In each iteration $j$ of this loop, each vertex $v\in H_{i}$ sequentially,
within $O(1)$ communication rounds, assigns a color to each edge
of $G_{j}(v)$ that has not yet been assigned to any edge incident
on it in $G$, from the palette $\{1,2,...,2\Delta-1\}$. Subsequently,
$v$ sends the colors assigned to edges incident on it, that have
already been colored, to all its neighbors in $V\setminus\cup_{k=1}^{i}H_{k}$.
Then, each neighbor $w\in V\setminus\cup_{k=1}^{i}H_{k}$ of $v$
marks the colors sent to it by $v$ as colors that can't be assigned
to edges that haven't yet been colored. The devised algorithm corresponds
to the following corollary.

\begin{corollary}\label{edgeColPoECorollary}

For an input graph $G=(V,E)$, it is possible to deterministically
compute a $(2\Delta-1)$-edge-coloring with a vertex-averaged complexity
of $O(a+\log^{*}n)$ rounds.

\end{corollary}

\begin{proof}

We first prove the correctness of the devised algorithm. Suppose that
we have invoked the devised algorithm $\mathcal{A}'$ on some graph
$G=(V,E)$. For some $1\le i\le\ell$, let us look at some $H$-set
$H_{i}$. For two edges $\{u,v\}$,$\{u,w\}$, such that $u,v,w\in H_{i}$,
by the correctness of the $\left(2\Delta-1\right)$-edge-coloring
algorithm of \cite{Panconesi2001}, $\{u,v\},\{u,w\}$ have different
colors.

Suppose now that $i\ge2$. For two edges $e=\{u,v\},e'=\{u,w\}$,
such that $u,v\in H_{i},w\in H_{k}$ for $k<i$, the vertex $u$ colored
$e'$ after $e$. Therefore, when $u$ colored $e'$, it made sure
to use an available color from the palette $\{1,2,...,2\Delta-1\}$,
that is different than that of $e$. Such a color necessarily exists
since each edge in $E$ can intersect at most $2\Delta-2$ other edges.
Therefore, $e,e'$ have different colors.

Let us look now at two edges $e=\{u,v\},e'=\{w,x\}$, such that $u,w\in H_{k}$,
for some $k<i$ and $v,x\in H_{i}$, and $e,e'$ have a common endpoint.
Suppose for contradiction that $e,e'$ have the same color. If $u=w$,
then $e,e'$ were assigned different labels in the step of $O(a)$-forests-decomposition
that algorithm $\mathcal{A}'$ carried out in the beginning of its
iteration $k$. Therefore, for some $j_{1}\ne j_{2},1\le j_{1},j_{2}\le A$,
in the loop of algorithm $\mathcal{B}$, executed in iteration $i$
of algorithm $\mathcal{A}'$, the edges $e,e'$ were colored in iterations
$j_{1},j_{2}$, respectively. Suppose without loss of generality that
$j_{2}>j_{1}$. Then the edge $e'$ was colored after $e$. Therefore,
$x$ assigned an available color to $e'$ from the palette $\{1,2,...,2\Delta-1\}$,
different than that of $e$. Therefore, $e,e'$ have actually been
assigned two different colors, leading to a contradiction.

Suppose now that $v=x$. If $e,e'$ have the same label $j$, then
$v$ assigned each edge $e,e'$ a different available color from the
palette $\{1,2,...,2\Delta-1\}$ in iteration $j$ of the loop of
algorithm $\mathcal{B}$. Suppose that the edges $e,e'$ have different
labels $j_{1},j_{2}$, respectively, and that without loss of generality
it holds that $j_{2}>j_{1}$. Then, $e'$ received a different color
from the palette $\{1,2,...,2\Delta-1\}$ than $e$ in iteration $j_{2}$
of the loop of algorithm $\mathcal{B}$, leading to a contradiction.
Therefore, for any two edges $e,e'\in E$, $e,e'$ received different
colors, completing the proof of correctness.

As for the number of colors used by the algorithm, each edge $e\in E$
is assigned some color from $\{1,2,...,2\Delta-1\}$. Therefore, the
number of colors employed by algorithm $\mathcal{A}'$ is at most
$2\Delta-1$.

We now analyze the vertex-averaged complexity of algorithm $\mathcal{A}'$.
The worst-case time complexity of algorithm $\mathcal{A}$ is asymptotically
equal to that of the $\left(2\Delta-1\right)$-edge-coloring algorithm
of \cite{Panconesi2001}, when executed on a sub-graph of $G$ induced
by an $H$-set $H_{i}$. This time complexity is $O\left(a+\log^{*}n\right)$.
In algorithm $\mathcal{B}$, we execute a loop consisting of $O(a)$
iterations, where each iteration takes constant time.

Therefore, by Theorem \ref{algorithmForPoEMethodTheorem}, the vertex-averaged
complexity of algorithm $\mathcal{A}'$ is $O(a+\log^{*}n)$, as required.

\end{proof}

Corollary \ref{edgeColPoECorollary} implies the following corollary
for graphs with a constant arboricity. 

\begin{corollary}\label{O(1)ArbEdgeColPoECorollary}

For an input graph $G=(V,E)$ with arboricity $a=O(1)$, one can deterministically
compute a $\left(2\Delta-1\right)$-edge-coloring of the edges of
$E$ with a vertex-averaged complexity of $O(\log^{*}n)$.

\end{corollary}

We now devise an algorithm for MM with improved vertex-averaged complexity.
We use the following algorithms as algorithms $\mathcal{A},\mathcal{B}$
in each iteration $i$ of $\mathcal{A}'$. We note that whenever an
edge $e=\{u,v\}$ is added to the matching, $u,v$ become inactive
and stop taking part in algorithm $\mathcal{A}'$.

In algorithm $\mathcal{A}$ we compute an MM for $G(H_{i})$ using
an algorithm from \cite{Panconesi2001}. While doing so, for $i\ge2$,
we make sure not to add edges that are incident on other edges, that
have already been added in iterations $k<i$. We execute algorithm
$\mathcal{B}$ if $i\ge2$. In algorithm $\mathcal{B}$, we handle
edges $e=\{u,v\}$, such that $u\in H_{k}$ for some $k<i$ and $v\in H_{i}$.
In algorithm $\mathcal{B}$, we execute a loop for $j=1,2,...,A$.
In each iteration $j$ of this loop, each vertex $v\in H_{i}$ selects
a single edge, from the described edges $e$ labeled $j$, that isn't
incident on another edge $e'$ belonging to the matching, if such
an edge exists, and adds it to the matching. The devised algorithm
corresponds to the following corollary.

\begin{corollary}\label{mmPOECorollary}

For an input graph $G=(V,E)$, a maximal matching can be computed
for the edges of $G$ with a vertex-averaged complexity of $O(a+\log^{*}n)$
rounds.

\end{corollary}

\begin{proof}

We first prove the correctness of algorithm $\mathcal{A}'$. We prove
this by induction on the number of the current iteration $i$ of algorithm
$\mathcal{A}'$. In the base case, we simply invoke the known algorithm
from \cite{Panconesi2001} on $G(H_{1})$. Therefore, by the correctness
of the MM algorithm of \cite{Panconesi2001}, algorithm $\mathcal{A}'$
correctly computes an MM for $G(H_{1})$.

Suppose now that algorithm $\mathcal{A}'$ correctly computes an MM
for $G(\cup_{l=1}^{k}H_{l})$ for some $1<k<\ell$. We now prove that
algorithm $\mathcal{A}'$ correctly computes an MM at the end of iteration
$i=k+1$ for $G(\cup_{l=1}^{k+1}H_{l})$. To prove this, we need to
prove that the matching computed so far, at the end of iteration $k+1$
of algorithm $\mathcal{A}'$, indeed contains no pair of intersecting
edges and that it is maximal for $G(\cup_{l=1}^{k+1}H_{l})$.

Let us look at some pair of intersecting edges $e=\{u,v\},e'=\{u,w\}$,
such that $e,e'\in E$, $u,v,w\in\cup_{l=1}^{k+1}H_{l}$. Following
a proof similar to that of the correctness of the algorithm devised
in Corollary \ref{edgeColPoECorollary}, we can show that if one of
the edges $e,e'$ belongs to the matching computed at the end of iteration
$i=k+1$, then the other does not. Therefore, the matching computed
indeed contains no pair of intersecting edges.

We now prove that the matching computed at the end of iteration $i=k+1$
of algorithm $\mathcal{A}'$ is maximal. Suppose for contradiction
that at the end of iteration $k+1$ of $\mathcal{A}'$ there exists
an edge $e=\{u,v\}\in E$, such that $u\in\cup_{l=1}^{k+1}H_{l},v\in H_{k+1}$,
which intersects no other edge $e'$ in the matching and does not
belong to the matching. If $u\in H_{k+1}$, then by the correctness
of the MM algorithm of \cite{Panconesi2001}, during the execution
of algorithm $\mathcal{A}$, $e$ would have been added to the matching,
leading to a contradiction. If $u\in H_{l}$, for some $l<k+1$, and
$e$ has some label $1\le j\le A$, then in iteration $j$ of the
loop of algorithm $\mathcal{B}$, $e$ should have been added to the
matching, also leading to a contradiction. Therefore, the matching
computed at the end of iteration $k+1$ of algorithm $\mathcal{A}'$
is maximal. It follows that the matching computed at the end of iteration
$k+1$ is indeed a maximal matching for $G(\cup_{l=1}^{k+1}H_{l})$.
Therefore, the matching algorithm $\mathcal{A}'$ computes for $G$
is indeed a maximal matching for $G$.

We now analyze the vertex-averaged complexity of algorithm $\mathcal{A}'$.
The worst-case time complexity of algorithm $\mathcal{A}$, where
we execute the MM algorithm of \cite{Panconesi2001}, following a
similar analysis to that of the proof of Corollary \ref{edgeColPoECorollary},
is $O(a+\log^{*}n)$. Also following a similar analysis to that of
the proof of Corollary \ref{edgeColPoECorollary}, the worst-case
time complexity of algorithm $\mathcal{B}$ is $O(a)$.

Therefore, by Theorem \ref{algorithmForPoEMethodTheorem}, the vertex-averaged
complexity of algorithm $\mathcal{A}'$ is $O(a+\log^{*}n)$, as required.

\end{proof}

Corollary \ref{mmPOECorollary} implies an additional corollary for
graphs with constant arboricity.

\begin{corollary}\label{O(1)ArbMMPOECorollary}

For an input graph $G=(V,E)$ with a known arboricity $a=O(1)$, an
MM can be computed within a vertex-averaged complexity of $O(\log^{*}n)$
rounds.

\end{corollary}

\section{Randomized Algorithms}

\subsection{Overview}

In this section we present randomized algorithms that achieve an improved
vertex-averaged complexity with high probability. In contrast to deterministic
algorithms, for which a guaranteed upper bound exists on the average
number of rounds per vertex required in any execution, it is not guaranteed
for any randomized algorithm that an upper bound for its vertex-averaged
complexity will hold for any single execution, but such an upper bound
can be shown to hold with high probability.

\subsection{A Randomized $(\Delta+1)$-Vertex-Coloring in $O(1)$ Vertex-Averaged
Complexity\label{subsec:randDeltaPlusOneColInO(1)AverageTime}}

In this section we analyze the vertex-averaged complexity of the randomized
algorithm for $(\Delta+1)$-vertex-coloring presented in Section 10.1
in \cite{Barenboim2013} as Procedure Rand-Delta-Plus1. This is a
variant of Luby's algorithm \cite{Luby1993}. In this algorithm, in
each round, each vertex $v\in V$ for an input graph $G=(V,E)$ first
draws a single bit $b$ from $\{0,1\}$ uniformly at random. If the
bit drawn is 0, the bit is discarded and then $v$ continues to the
next round. Otherwise, $v$ draws a color uniformly at random from
the set $\{1,2,...,\Delta+1\}\setminus F_{v}$ where $F_{v}$ is the
set of final colors selected by neighbors of $v$. If the color chosen
by $v$ is different than the color chosen by each of its neighbors
which also chose a color in the current round, and also different
than that of each of its neighbors which decided upon a final color
in a previous round, then $v$ declares the color chosen to be its
``final color'', sending a message to its neighbors announcing it.
Then, every neighbor $u$ of $v$ that has not decided upon a final
color yet, updates its list of ``final colors'' $F_{u}$ accordingly,
adding the color chosen by $v$. We present and prove the following
claim about the algorithm's vertex-averaged complexity.

\begin{theorem}\label{randDeltaPlusOneColTheorem}

For an input graph $G=(V,E)$ with maximum degree $\Delta$ it is
possible to compute a $(\Delta+1)$-vertex-coloring of $G$ within
a vertex-averaged complexity of $O(1)$ rounds, with high probability.

\end{theorem}

\begin{proof}

It it shown in \cite{Barenboim2013} that the probability for a vertex
$v$ to terminate in a given round is at least $\frac{1}{4}$. The
explanation for this is as follows. Suppose that in some round $i$,
in the first step of the algorithm, some vertex $v\in V$, chose a
bit $1$ from $\{0,1\}$ uniformly at random. Therefore, in the second
step of the algorithm, it chose a color uniformly at random from $\left\{ 1,2,...,\Delta+1\right\} \setminus F_{v}$.
For some neighbor $u$ of $v$, the probability that $u$ chose some
color in the current round and that $v$ chose the same color as $u$
is at most:
\[
\frac{1}{2(\Delta+1-\left|F_{v}\right|)}
\]
This follows from the following observations. The probability for
$u$ not to discard the bit $b\in\{0,1\}$ it initially selected uniformly
at random is $\frac{1}{2}$. In addition, $v$ has at most $\Delta+1-\left|F_{v}\right|$
different colors to choose from.

Therefore, using the union bound, the probability for $v$ to select
in a given round a color identical to the color chosen by any of its
neighbors is at most $\frac{(\Delta+1-\left|F_{v}\right|)}{2(\Delta+1-\left|F_{v}\right|)}=\frac{1}{2}$.
Also, the probability for $v$ not to discard its bit $b\in\{0,1\}$
and proceed to select uniformly at random a color from $\{1,2,...,\Delta+1\}\setminus F_{v}$
is $\frac{1}{2}$. It follows that the probability for $v$ to select
a color different from that of each of its neighbors is at least $\frac{1}{4}$,
as required.

Therefore, the probability that a vertex $v$ does not terminate in
a given round is at most $\frac{3}{4}$. Let us denote the number
of vertices of $V$ that have not terminated in a given round $i$
as $X_{i}$. Also, for a vertex $v$ and a round $i$, let us denote
by $X_{v,i}$ the following random variable. The random variable $X_{v,i}$
receives the value 1 if $v$ did not choose a final color and did
not terminate in round $i$, on the condition that $v$ did not terminate
in any round $j<i$. Otherwise, if $v$ terminates in round $i$,
it holds that $X_{v,i}=0$. By the linearity of expectation, it follows
that $E[X_{i}]\le\frac{3}{4}n_{i}$ where $n_{i}$ is the number of
active vertices in round $i$, that is vertices that have not yet
decided upon a final color.

Now, let us denote the vertices of $V$ as $v_{1},v_{2},...,v_{n}$.
Also, for each random variable $X_{v,i}$, for a vertex $v\in V$,
we define a twin random variable $\hat{X}_{v,i}$, such that the resulting
random variables $\hat{X}_{v,i}$ satisfy the following properties.
For each $v\in V$, the random variables $X_{v,i},\hat{X}_{v,i}$
have the same distribution, and the random variables $\hat{X}_{v_{1},i},\hat{X}_{v_{2},i},...,\hat{X}_{v_{n},i}$
are independent. In addition, we denote $\hat{X}_{i}=\sum_{j=1}^{n}\hat{X}_{v_{j},i}$.
Therefore, by applying the Chernoff bound, we obtain the following:

\begin{align*}
Pr\left[X_{i}\ge\frac{15}{16}n_{i}\right] & \le Pr\left[X_{i}\ge\left(1+\frac{1}{4}\right)E[X_{i}]\right]\\
 & =Pr\left[X_{i}\ge\left(1+\frac{1}{4}\right)\sum_{j=1}^{n}E[X_{v_{j},i}]\right]\\
 & =Pr\left[X_{i}\ge\left(1+\frac{1}{4}\right)\sum_{j=1}^{n}E\left[\hat{X}_{v_{j},i}\right]\right]\\
 & =Pr\left[X_{i}\ge\left(1+\frac{1}{4}\right)E[\hat{X}_{i}]\right]\\
 & \le e^{-\frac{\left(\frac{1}{4}\right)^{2}\cdot\frac{3}{4}n_{i}}{3}}\\
 & =e^{-\frac{n_{i}}{64}}
\end{align*}
Also, because in each round a vertex does not terminate with probability
at most $\frac{3}{4}$, the probability for any vertex not to terminate
after $c'\log n$ rounds, for an appropriate constant $c'\ge5$, by
the union bound, is at most:
\[
n\cdot\left(\frac{3}{4}\right)^{c'\log n}=\frac{n}{n^{c'\log\frac{4}{3}}}<\frac{n}{n^{2}}=\frac{1}{n}
\]
Therefore, either for the current round $i$ the value of $n_{i}$
satisfies $\frac{n_{i}}{64}\ge\log n^{64c}$ for a constant $c>1$,
in which case the number of vertices that have not yet terminated
becomes smaller by a constant fraction with high probability, or $n_{i}<4,096c\cdot\log n$,
in which case, each vertex that has not yet terminated in this round
will terminate after $O(\log n)$ more rounds with high probability.
Let $k$ denote the number of the first round in which $n_{k}<4,096c\cdot\log n$.
Then the total number of rounds of computation executed by all vertices
is with high probability at most:
\[
\sum_{i=1}^{k-1}\left(\frac{3}{4}\right)^{i}\cdot n+O\left(\log^{2}n\right)=O(n)+O\left(\log^{2}n\right)=O(n)
\]
Therefore, the vertex-averaged complexity of the devised algorithm
is $O(1)$. In terms of the number of colors employed by the algorithm,
it remains at most $\Delta+1$, as the final color of each vertex
$v\in V$ is still only one color $c_{v}\in\{1,2,...,\Delta+1\}$.

\end{proof}

Therefore, the vertex-averaged complexity of this known algorithm
is $O(1)$ rounds only. Also the achieved coloring is proper as we
haven't changed the original algorithm in this analysis in any way.
In addition, the achieved coloring employs at most $\Delta+1$ colors
as is already presented in \cite{Barenboim2013}. This result is far
superior to any currently known worst-case time result for $(\Delta+1)$-vertex-coloring
according to \cite{Barenboim2016b}. 

We also remind that a lower bound of $\Omega\left(\log^{*}n\right)$
rounds was presented in \cite{Feuilloley2017} for deterministic $3$-vertex-coloring
of cycles. Cycles satisfy $\Delta=2$. Therefore, this lower bound
holds for deterministic $\left(\Delta+1\right)$-vertex-coloring for
general graphs. On the other hand, by Theorem \ref{randDeltaPlusOneColTheorem},
using randomization helps obtain an algorithm with a vertex-averaged
complexity that isn't bounded from below by $\Omega\left(\log^{*}n\right)$,
and that is even constant. A similar, weaker result was already presented
in \cite{Feuilloley2017}. In the latter case, the author showed that
in the case of $3$-coloring cycles, using the randomized algorithm
described in this section, for $\Delta=2$, the expected vertex-averaged
complexity was constant. The result of the analysis in this section,
however, is stronger, since it applies to a general value of $\Delta$,
and since it shows the time complexity of the algorithm described
in this section to be constant with high probability, rather than
only show the expected value of the time complexity is constant.

\subsection{Additional Randomized Vertex-Coloring Algorithms Using $O(1)$ Vertex-Averaged
Complexity\label{subsec:randO(aloglogn)ColoringInO(1)Time}}

In this section, we devise a new randomized algorithm based on the
randomized $\left(\Delta+1\right)$-vertex-coloring algorithm of Section
\ref{subsec:randDeltaPlusOneColInO(1)AverageTime}. The devised algorithm
consists of two phases, and proceeds as follows.

In the first phase, we execute $t=\left\lfloor 2\log\log n\right\rfloor $
iterations of Procedure Partition, thus forming the $H$-sets $H_{1},H_{2},...,H_{t}$.
Upon the formation of each $H$-set $H_{i}$, for $1\le i\le t$,
we invoke the randomized $\left(\Delta+1\right)$-vertex-coloring
algorithm of Section \ref{subsec:randDeltaPlusOneColInO(1)AverageTime}
on $G\left(H_{i}\right)$, with the palette $\{1,2,...,A+1\}$. Once
a vertex selects a final color $c_{v}$ from this palette, according
to the execution on $G(H_{i})$, we assign each vertex $v\in H_{i}$
the color $\left\langle c_{v},i\right\rangle $.

In the second phase, that starts in round $t+1$, we continue executing
Procedure Partition, in order to compute the sets $H_{t+1},H_{t+2},...H_{\ell}$
(possibly in parallel to actions of some vertices in $H_{1},H_{2},...,H_{t}$
that still have not terminated phase 1). Then, we perform a loop from
$j=\ell$ down to $t+1$. In each iteration $j$ of this loop, we
invoke the algorithm of Section \ref{subsec:randDeltaPlusOneColInO(1)AverageTime}
with the palette $\{A+2,...,2A+2\}$. We note that for each iteration
$j<\ell$, each vertex $v\in H_{j}$ first waits for all its neighbors
in $\cup_{l=j+1}^{\ell}H_{l}$ to terminate and choose a color from
the palette $\{A+2,A+3,...,2A+2\}$, and only then invokes the algorithm
of Section \ref{subsec:randDeltaPlusOneColInO(1)AverageTime}. The
properties of the devised algorithm are summarized in the following
theorem.

\begin{theorem}\label{randO(a)ColInO(loglogn)AverageTimeTheorem}

The devised algorithm colors an input graph $G=(V,E)$ with a known
arboricity $a$ using $O(a\log\log n)$ colors with a vertex-averaged
complexity of $O(1)$ rounds, with high probability.

\end{theorem}

\begin{proof}

First, we prove the algorithm computes a proper coloring of the input
graph. According to Section \ref{subsec:Procedure-Partition-Description},
the maximum degree of each $H$-set $H_{i}$, for $1\le i\le\ell$,
is $\Delta\left(G\left(H_{i}\right)\right)=A=O(a)$. We invoked the
algorithm of Section \ref{subsec:randDeltaPlusOneColInO(1)AverageTime}
with a palette of size $A+1$ on $G\left(H_{i}\right)$, for each
$1\le i\le\ell$. Therefore, for any two vertices $u,v\in H_{i}$,
for $1\le i\le\ell$, by the correctness of the algorithm of Section
\ref{subsec:randDeltaPlusOneColInO(1)AverageTime}, $u,v$ have different
colors, as required. For two vertices $u\in H_{j},v\in H_{i}$, such
that $1\le j,i\le t$, $j<i$, they have each respectively been assigned
the colors $\left\langle c_{u},j\right\rangle ,\left\langle c_{v},i\right\rangle $.
Since the color of the vertices $u,v$ differs at least in one component,
the second component, they have different colors.

Let us now look at two vertices $u\in H_{j},v\in H_{i}$, such that
$t+1\le i,j\le\ell$, $j<i$. In iteration $l=j$ of the loop of the
second phase of the devised algorithm, we invoked the algorithm of
Section \ref{subsec:randDeltaPlusOneColInO(1)AverageTime} with the
palette $\{A+2,A+3,...,2A+2\}$. By construction of the devised algorithm,
the vertex $u$ was colored after $v$, using the standard randomized
$\left(\Delta+1\right)$-vertex-coloring algorithm of Section \ref{subsec:randDeltaPlusOneColInO(1)AverageTime},
while avoiding the use of colors already chosen by neighbors of $u$
in $V\setminus\left(\cup_{l=1}^{j}H_{l}\right)$. Therefore, if $u$
has an available color to choose from the respective palette, it is
different than the color of each neighbor of $u$ in $V\setminus\left(\cup_{l=1}^{j}H_{l}\right)$.
The vertex $u$ executes the standard randomized $\left(\Delta+1\right)$-vertex-coloring
algorithm of Section \ref{subsec:randDeltaPlusOneColInO(1)AverageTime},
using a palette of colors of size $A+1$, while $u$ has at most $A$
neighbors in $V\setminus\left(\cup_{l=1}^{j}H_{l}\right)$, according
to Section \ref{subsec:The-Partition-Algorithm}. Therefore, $u$
has at least one available color to choose from, when it executes
the standard randomized $\left(\Delta+1\right)$-vertex-coloring algorithm
of Section \ref{subsec:randDeltaPlusOneColInO(1)AverageTime}. Therefore,
for any two vertices $u\in H_{j},v\in H_{i}$, for $t+1\le i,j\le\ell$,
$j<i$, the vertices $u,v$ have different colors.

For any two vertices $u\in\cup_{i=1}^{t}H_{i}$, $v\in\cup_{i=t+1}^{\ell}H_{i}$,
the vertices $u,v$ each respectively either chose a color $\left\langle c_{u},i\right\rangle $,
where $c_{u}\in\{1,2,...,A+1\}$, or a color from the palette $\{A+2,A+3,...,2A+2\}$.
Since the set from which the vertex $u$ chose its color, and the
set from which $v$ chose its color, are disjoint, the vertices $u,v$
necessarily have different colors. Therefore, for any two vertices
$u,v\in V$, it holds that $u,v$ have different colors, as required.

We now analyze the number of colors employed by the devised algorithm.
In the first phase, each vertex $v\in H_{i}$, for $1\le i\le t$,
chooses a color $\left\langle c_{v},i\right\rangle $, where $c_{v}\in\{1,2,...,A+1\}$
and $t=\left\lfloor 2\log\log n\right\rfloor $. In the second phase,
each vertex $v\in\cup_{i=t+1}^{\ell}H_{i}$ chooses a color from the
set $\{A+2,A+3,...,2A+2\}$. Therefore, the total number of colors
used by the devised algorithm is $O\left(a\log\log n\right)$.

We now analyze the vertex-averaged complexity of the devised algorithm.
We want to show that the total number of communication rounds carried
out by all vertices of $V$ is $O(n)$ with high probability. First,
we analyze the total number of rounds of communication carried out
in the first phase of the devised algorithm.

Denote by $n_{i}$ the number of active vertices in the input graph
that haven't yet finished all computation at the beginning of iteration
$i=1,2,...,t$ of the first phase. Also, suppose that the execution
of the algorithm of Section \ref{subsec:randDeltaPlusOneColInO(1)AverageTime}
on $G\left(H_{i}\right)$ begins in round $r_{i}$. Then, denote by
$n_{i,j}$ the number of vertices in $H_{i}$ that remain uncolored
by the invocation of the algorithm of Section \ref{subsec:randDeltaPlusOneColInO(1)AverageTime}
in round $j\ge r_{i}$. Also, let $j=k$ be the first round where
$n_{i,j}<4,096c\cdot\log n$ for a constant $c>1$. Then, according
to the proof of Theorem \ref{randDeltaPlusOneColTheorem}, the total
number of rounds of communication executed by the vertices of a given
subset $H_{i}$ in the execution of the algorithm of Section \ref{subsec:randDeltaPlusOneColInO(1)AverageTime}
is, with high probability, at most:
\[
\sum_{j=1}^{k-1}\left(\frac{3}{4}\right)^{j}\cdot n_{i,j}+O\left(\log n\right)=O(n_{i})+O\left(\log n\right)
\]
It follows that with high probability, the total number of rounds
of communication carried out by the vertices of $\cup_{i=1}^{t}H_{i}$
is:
\[
\sum_{i=1}^{t}O\left(n_{i}\right)+O\left(t\log n\right)=O\left(n+\log n\log\log n\right)=O(n)
\]
We now analyze the total number of rounds of communication carried
out in the second phase of the devised algorithm. Let us look at some
vertex $v\in H_{i}$, for $t+1\le i\le\ell$. For each two different
rounds $j,k$ of the algorithm of Section \ref{subsec:randDeltaPlusOneColInO(1)AverageTime}
executed by the vertex $v$, such that $j<k$, the color chosen by
$v$ in round $k$ is independent of the color it chose in round $j$.
Also, according to Section \ref{subsec:randDeltaPlusOneColInO(1)AverageTime},
in each round, in the execution of the algorithm of Section \ref{subsec:randDeltaPlusOneColInO(1)AverageTime},
the vertex $v$ has a probability of at most $\frac{3}{4}$ to choose
a color identical to the color chosen by one of its neighbors, and
thus, not to terminate. Therefore, the probability that $v$ doesn't
terminate after $l$ rounds is at most $\left(\frac{3}{4}\right)^{l}$.

Throughout the execution of the second phase of the devised algorithm,
there can trivially be no more than $n$ active vertices in any round.
Therefore, by the union bound, the probability that in the execution
of the algorithm of Section \ref{subsec:randDeltaPlusOneColInO(1)AverageTime},
in an iteration $t+1\le i\le\ell$ of the loop of the second phase
of the devised algorithm, at least one vertex $v\in H_{i}$ doesn't
terminate after $l$ rounds is at most:
\begin{align*}
\left(\frac{3}{4}\right)^{l}\cdot n & =2^{l\cdot\log\frac{3}{4}}\cdot n\\
 & =\frac{n}{2^{l\cdot\log\frac{4}{3}}}
\end{align*}
Specifically, for $l=\left(c_{i}+1+\frac{2}{\epsilon}\right)\cdot\frac{1}{\log\frac{4}{3}}\cdot\log n$
for a sufficiently large constant $c_{i}>1$, the probability of the
last mentioned event is at most $\frac{n}{n^{c_{i}+1+\frac{2}{\epsilon}}}=\frac{1}{n^{c_{i}+\frac{2}{\epsilon}}}$.
We remind, that according to Section \ref{subsec:Procedure-Partition-Description},
it holds that $\ell=\lfloor\frac{2}{\epsilon}\log n\rfloor\le\frac{2}{\epsilon}\log n$.
Choosing an appropriate constant $c_{min}=\min_{t+1\le i\le\ell}c_{i}$,
by the union bound, the probability that in at least one iteration
$t+1\le i\le\ell$ of the second phase of the devised algorithm, the
running time of the invocation of the algorithm of Section \ref{subsec:randDeltaPlusOneColInO(1)AverageTime}
requires more than $\left(c_{i}+1+\frac{2}{\epsilon}\right)\cdot\frac{1}{\log\frac{4}{3}}\cdot\log n$
rounds, is at most:
\[
\frac{\ell}{n^{c_{min}+\frac{2}{\epsilon}}}\le\frac{\frac{2}{\epsilon}\log n}{n^{c_{min}+\frac{2}{\epsilon}}}\le\frac{n^{\frac{2}{\epsilon}}}{n^{c_{min}+\frac{2}{\epsilon}}}=\frac{1}{n^{c_{min}}}<\frac{1}{n}
\]
Therefore, with high probability, the worst-case time complexity of
the second phase of the devised algorithm is $O\left(\log^{2}n\right)$.
Also, according to Lemma \ref{activeVerticesNumberLemma}, the number
of active vertices in the execution of the second phase of the devised
algorithm is $O\left(\frac{n}{\log^{2}n}\right)$. Therefore, the
total number of communication rounds carried out in the second phase
of the devised algorithm is $O(n)$, with high probability. 

Therefore, for a suitable choice of the constant $c'$ from the proof
of Theorem \ref{randDeltaPlusOneColTheorem}, and of the constants
$c_{i}$, for $t+1\le i\le\ell$, the total number of communication
rounds carried out by the vertices of the input graph in the execution
of the devised algorithm is $O(n)$, with high probability. Therefore,
the vertex-averaged complexity of the devised algorithm is $O(1)$
rounds, with high probability, as required.

\end{proof}

\section{Discussion\label{sec:discussionAndConclusionsSection}}

Based on the results presented so far, we make several observations.
We remind that a summary of the results obtained in the paper can
be found in Section \ref{sec:comparisonWithPreviousWorkSection}.
One observation we make pertains to differences obtained so far, between
the lower and upper bounds on the vertex-averaged complexity of the
solution of the problems studied in this paper. For example, let us
look at the currently-known best lower bound on the vertex-averaged
complexity of deterministic $O\left(a\cdot\log^{*}n\right)$-vertex-coloring.
In this case, the only known non-trivial lower bound on the vertex-averaged
complexity is $\Omega\left(1\right)$. For constant arboricity, our
results are away of this bound only by a factor of $\log^{*}n$. For
non-constant arboricity, there exists a difference between the worst-case
lower bound of $\Omega\left(1\right)$ and the vertex-averaged time
complexity of the respective algorithm we devised, which is $O\left(a\log^{*}n\right)$.
An interesting question these observations raise, is whether the lower
bound for deterministic $O\left(a\cdot\log^{*}n\right)$-vertex-coloring
can be shown to be $\omega\left(1\right)$, and in particular, how
close it is to $O\left(a\log^{*}n\right)$ rounds.

As another example, in the case of randomized $p$-vertex-coloring,
for $p\in\left\{ O\left(a\log\log n\right),\Delta+1\right\} $, we
have obtained the best possible results (up to constant factors).
Namely, for $p\in\left\{ O\left(a\log\log n\right),\Delta+1\right\} $,
we have presented methods for coloring an input graph with a vertex-averaged
complexity of $O\left(1\right)$ rounds, with high probability.

Another observation we make is about the difference between the best
vertex-averaged complexity obtained using deterministic algorithms,
and the best vertex-averaged complexity obtained using randomized
algorithms. This observation concerns only the vertex-coloring algorithms
presented in this paper, as the algorithms presented for MIS, MM and
$\left(2\Delta+1\right)$-edge-coloring are all deterministic. Consider
for example the case of $O\left(a\log\log n\right)$-vertex-coloring.
In this case, the vertex-averaged time complexity of the deterministic
algorithm we devised in Section \ref{subsec:O(ka)VertexColoringSection}
is $O\left(a\log\log n\right)$, according to Theorem \ref{O(ak)DeterministicVertexColoringTheorem}.
On the other hand, the randomized algorithm of Section \ref{subsec:randO(aloglogn)ColoringInO(1)Time}
has a vertex-averaged complexity of only $O\left(1\right)$ rounds,
with high probability. Therefore, there is a difference in this case
of $O\left(a\log\log n\right)$ rounds between the vertex-averaged
complexity of the deterministic algorithm and the randomized algorithm.

This raises the question whether there exists a deterministic $O\left(a\log\log n\right)$-vertex-coloring
algorithm with a vertex-averaged complexity that is $o\left(a\log\log n\right)$,
and how much of an improvement can be obtained.

In the case of $\left(\Delta+1\right)$-vertex-coloring, the deterministic
algorithm of Corollary \ref{vertexColPoECorollary} requires a vertex-averaged
complexity of $O(\sqrt{a}\log^{2.5}a+\log^{*}n)$ rounds, which depends
on both $a$ and $n$. On the other hand, the standard randomized
$\left(\Delta+1\right)$-vertex-coloring algorithm, described in Section
\ref{subsec:randDeltaPlusOneColInO(1)AverageTime}, has a vertex-averaged
complexity of only $O\left(1\right)$ rounds, with high probability.
This last mentioned difference between the last two mentioned algorithms
raises the question to what extent can the vertex-averaged complexity
of the deterministic algorithm of Corollary \ref{vertexColPoECorollary}
be made closer to the constant vertex-averaged complexity of the standard
randomized $\left(\Delta+1\right)$-vertex-coloring algorithm, described
in Section \ref{subsec:randDeltaPlusOneColInO(1)AverageTime}.

Yet an additional observation that can be made regarding the results
obtained in this paper, relates to the difference between the known
worst-case lower bounds for the problems that were dealt with in this
paper, and certain results obtained in this paper. Namely, the vertex-averaged
complexity of some of the algorithms described in this paper, is asymptotically
lower than the worst-case lower bound for the same respective problems.
Following are some examples of this observation.

In the case of $O\left(a\cdot\log^{*}n\right)$-vertex-coloring, for
constant arboricity $a$, it is implied in \cite{Barenboim2008} that
solving this problem requires $\Omega\left(\log n\right)$ time in
the worst case. However, the deterministic algorithm of Section \ref{subsec:O(ka)VertexColoringSection}
obtains a vertex-averaged complexity of $O\left(\log^{*}n\right)$.

In the case of $\left(\Delta+1\right)$-vertex-coloring, we have shown
in Section \ref{subsec:randDeltaPlusOneColInO(1)AverageTime} that
the standard randomized $\left(\Delta+1\right)$-vertex-coloring algorithm
described in the same section has a vertex-averaged complexity of
$O\left(1\right)$ rounds, with high probability, while in the worst-case
scenario, according to \cite{Linial1992}, such a coloring requires
$\Omega\left(\log^{*}n\right)$ rounds, even with the use of randomization.

In the case of $O\left(a\cdot\log^{*}n\right)$-vertex-coloring, for
constant arboricity, a deterministic algorithm we presented in Section
\ref{subsec:O(ka)VertexColoringSection} produces the coloring with
a vertex-averaged complexity of $O\left(a\log^{*}n\right)=O\left(\log^{*}n\right)$
rounds. This last mentioned vertex-averaged complexity is considerably
lower than the worst-case lower bound for for $O(aq)$-coloring, $q>2$,
which is, for general arboricity, $\Omega\left(\frac{\log n}{\log a+\log q}\right)$
\cite{Barenboim2008}. In particular, for $a=O\left(\log^{*}n\right)$,
the worst-case lower bound for $O\left(a\log^{*}n\right)$ is $\Omega\left(\frac{\log n}{\log\left(\log^{*}n\right)}\right)$.

\section{Conclusion and future research directions}

We observe that from a theoretical point of view, the measure of vertex-averaged
complexity seems promising as an optimization criterion. In particular,
in many cases where a certain pair of algorithms have the same asymptotic
worst case running time, the measure of vertex-averaged complexity
provides a way to decide which algorithm is expected to yield better
performance in practice, especially in scenarios such as those mentioned
in Section \ref{subsec:Motivation}.

Regarding possible directions for future research on the subject of
vertex-averaged complexity, one direction is to perform an experimental
evaluation of the value of using vertex-averaged complexity as an
optimization criterion. The goal of such experimentation would be
to confirm that the framework obtains simulations with a significantly
smaller actual running time (and equivalently, also, a significantly
smaller total number of rounds over all simulated processors).

Another possible research direction, is an empirical evaluation of
the extent to which algorithms with improved vertex-averaged complexity,
for tasks consisting of at least two sub-tasks, allow a faster execution
of the complete task for most of the simulated network's processors.

An additional research direction is to try and improve upon the results
obtained in this paper. In particular, one can try and work on developing
algorithms with better vertex-averaged complexity for the problems
discussed in the paper, as well as possibly other important graph-theoretic
problems.

In addition, one can try and research regarding possible non-trivial
(in particular, non constant) lower bounds on the vertex-averaged
complexity of the symmetry-breaking problems discussed in this paper,
as well as the vertex-averaged complexity of possibly other important
graph-theoretic problems. Currently, the only known non-trivial lower
bound is the lower bound of $\Omega\left(\log^{*}n\right)$ rounds
for the deterministic $3$-coloring of rings (and therefore, for many
other types of coloring problems, as well), presented in \cite{Feuilloley2017}.

\end{onehalfspacing}


\begin{thebibliography}{10}
\bibitem{Awerbuch1989}B. Awerbuch, M. Luby, A.V. Goldberg and S.A.
Plotkin. Network Decomposition and Locality in Distributed Computation.
In \emph{Proc. of the 30th Annual Symposium on Foundations of Computer
Science}, pages 364-369, 1989. 

\bibitem{Barenboim2017}L. Barenboim and T. Maimon. Fully-Dynamic
Graph Algorithms with Sublinear Time Inspired by Distributed Computing.
\emph{Procedia Computer Science}, 108:89-98, 2017.

\bibitem{Barenboim2016b}L. Barenboim. Deterministic ($\Delta$ +
1)-Coloring in Sublinear (in $\Delta$) Time in Static, Dynamic, and
Faulty Networks. \emph{Journal of the ACM}, 63(5):47:1\textendash 47:22,
2016. 

\bibitem{Barenboim2013}L. Barenboim and M. Elkin. \emph{Distributed
Graph Coloring: Fundamentals and Recent Developments}. Morgan \& Claypool
Publishers, San Francisco, CA, 2013. 

\bibitem{Barenboim2011}L. Barenboim and M. Elkin. Deterministic Distributed
Vertex Coloring in Polylogarithmic Time. \emph{Journal of the ACM},
58(5):23:1\textendash 23:25, 2011. 

\bibitem{Barenboim2010}L. Barenboim and M. Elkin. Sublogarithmic
distributed MIS algorithm for sparse graphs using Nash-Williams decomposition.
\emph{Distributed Computing}, 22(5):363\textendash 379, 2010. 

\bibitem{Barenboim2009}L. Barenboim and M. Elkin. Distributed ($\Delta$+1)-coloring
in Linear (in $\Delta$) Time. In \emph{Proc. of the 41st Annual ACM
Symposium on Theory of Computing}, pages 111-120, 2009. 

\bibitem{Barenboim2008}L. Barenboim and M. Elkin. Sublogarithmic
Distributed MIS Algorithm for Sparse Graphs Using Nash-williams Decomposition.
In \emph{Proc. of the 27th ACM Symposium on Principles of Distributed
Computing}, pages 25-34, 2008. 

\bibitem{Chang2018}Y. J. Chang, W. Li and S. Pettie. An Optimal Distributed
$\left(\Delta+1\right)$-Coloring Algorithm? In \emph{Proceedings
of the 50th Annual ACM SIGACT Symposium on Theory of Computing}, pages
445-456, 2018.

\bibitem{Cole1986}R. Cole and U. Vishkin. Deterministic Coin Tossing
and Accelerating Cascades: Micro and Macro Techniques for Designing
Parallel Algorithms. In \emph{Proc. of the 18th Annual ACM Symposium
on Theory of Computing}, pages 206-219, 1986. 

\bibitem{Feuilloley2017a}L. Feuilloley. How Long It Takes for an
Ordinary Node with an Ordinary ID to Output? \emph{arXiv.org}. Retrieved
from https://arxiv.org/abs/1704.05739, 21 November 2017. Accessed
13 October 2018.

\bibitem{Feuilloley2017}L. Feuilloley. How Long It Takes for an Ordinary
Node with an Ordinary ID to Output? In \emph{Proceedings of the 24th
International Colloquium on Structural Information and Communication
Complexity (SIROCCO 2017)}, pages 263-282, 2017.

\bibitem{Fraigniaud2016}P. Fraigniaud, M. Heinrich and A. Kosowski.
Local Conflict Coloring. In \emph{Proceedings of} \emph{the 2016 IEEE
57th Annual Symposium on Foundations of Computer Science (FOCS)},
pages 625-634, 2016.

\bibitem{Goldberg1988}A.V. Goldberg, S.A. Plotkin and G.E. Shannon.
Parallel Symmetry-Breaking in Sparse Graphs. \emph{SIAM Journal on
Discrete Mathematics}, 1(4):434-446, 1988. 

\bibitem{Israel1986}A. Israel and A. Itai. A Fast and Simple Randomized
Parallel Algorithm for Maximal Matching. \emph{Information Processing
Letters}, 22(2):77\textendash 80, 1986. 

\bibitem{Israeli1986}A. Israeli and Y. Shiloach. An Improved Parallel
Algorithm for Maximal Matching. \emph{Information Processing Letters},
22(2):57\textendash 60, 1986. 

\bibitem{Kutten1995}S. Kutten and D. Peleg. Fault-local Distributed
Mending (Extended Abstract). In \emph{Proceedings of the Fourteenth
Annual ACM Symposium on Principles of Distributed Computing}, pages
20-27, 1995.

\bibitem{Lenzen2011}C. Lenzen and R. Wattenhofer. MIS on Trees. In
\emph{Proceedings of the 30th Annual ACM SIGACT-SIGOPS Symposium on
Principles of Distributed Computing}, pages 41-48, 2011.

\bibitem{Linial1992}N. Linial. Locality in Distributed Graph Algorithms.
\emph{SIAM Journal on Computing}, 21(1):193-201, 1992. 

\bibitem{Linial1987}N. Linial. Distributive Graph Algorithms Global
Solutions from Local Data. In\emph{ Proc. of the 28th Annual Symposium
on Foundations of Computer Science}, pages 331-335, 1987.

\bibitem{Luby1993}M. Luby. Removing randomness in parallel computation
without a processor penalty. \emph{Journal of Computer and System
Sciences}, 47(2):250-286, 1993.

\bibitem{Luby1986}M. Luby. A Simple Parallel Algorithm for the Maximal
Independent Set Problem. \emph{SIAM Journal on Computing}, 15(4):1036-1053,
1986. 

\bibitem{Naor1995}M. Naor and L. Stockmeyer. What Can be Computed
Locally? \emph{SIAM Journal on Computing, }24(6):1259-1277, 1995.

\bibitem{Panconesi2001}A. Panconesi and R. Rizzi. Some Simple Distributed
Algorithms for Sparse Networks. \emph{Distributed Computing}, 14(2):97\textendash 100,
2001.

\bibitem{Panconesi1996}A. Panconesi and A. Srinivasan. On the Complexity
of Distributed Network Decomposition. \emph{Journal of Algorithms},
20(2):356\textendash 374, 1996.

\bibitem{Parter2016}M. Parter, D. Peleg and S. Solomon. Local-on-average
Distributed Tasks. In \emph{Proc. of the 27th Annual ACM-SIAM Symposium
on Discrete Algorithms}, pages 220-239, 2016.
\end{thebibliography}
\end{document}